\documentclass[11pt]{article}

\usepackage{parskip}
\usepackage{amsmath}
\usepackage{amssymb}
\usepackage{xspace}
\usepackage{xurl}
\usepackage{tikz}

\newcounter{theorem} 
\newtheorem{xdefinition}[theorem]{Definition}
\newtheorem{xobservation}[theorem]{Observation}
\newtheorem{xtheorem}[theorem]{Theorem}
\newtheorem{xlemma}[theorem]{Lemma}	
\newtheorem{xproposition}[theorem]{Proposition}
\newtheorem{xcorollary}[theorem]{Corollary}
\newenvironment{definition}{\begin{xdefinition}\rm}%
{\hspace*{\fill}\raisebox{-1pt}{\boldmath$\Box$}\end{xdefinition}}
\newenvironment{observation}{\begin{xobservation}\rm}%
{\hspace*{\fill}\raisebox{-1pt}{\boldmath$\Box$}\end{xobservation}}
\newenvironment{theorem}{\begin{xtheorem}\rm}{\end{xtheorem}}
\newenvironment{lemma}{\begin{xlemma}\rm}{\end{xlemma}}
\newenvironment{proposition}{\begin{xproposition}\rm}{\end{xproposition}}
\newenvironment{corollary}{\begin{xcorollary}\rm}{\end{xcorollary}}
\newenvironment{proof}{\begin{trivlist}\item[]{\bf Proof }}%
{\hspace*{\fill}\raisebox{-1pt}{\boldmath$\Box$}\end{trivlist}}

\usepackage{enumitem}
\setlist[itemize,2]{label={$\circ$}}

\usepackage[T1]{fontenc}
\usepackage{dsfont}

\usepackage{xcolor}
\definecolor{custom_yellow}{RGB}{246,219,166}
\definecolor{custom_yellow_dark}{RGB}{196,136,19}
\definecolor{custom_blue}{RGB}{166,193,246}
\definecolor{custom_blue_dark}{RGB}{19,79,196}
\definecolor{custom_red}{RGB}{244,123,126}
\definecolor{custom_red_dark}{RGB}{220,18,23}

\usepackage{enumitem}

\usepackage{algorithm}
\usepackage[noend]{algpseudocode}

\usetikzlibrary{arrows.meta}
\usepackage{chemarrow}

\newcommand{\on}[1]{\ensuremath{\operatorname{#1}}}
\newcommand{\abs}[1]{\ensuremath{\left\lvert #1 \right\rvert}}
\newcommand{\eps}{\ensuremath{\varepsilon}}

\newcommand{\AT}{\ensuremath{\kappa}} 

\newcommand{\xZero}[1]{\ensuremath{#1_0}\xspace}
\newcommand{\xOne}[1]{\ensuremath{#1_1}\xspace}

\newcommand{\ALGS}[1]{\ensuremath{\mathcal{A}_{#1}}\xspace}
\newcommand{\ORALG}[1]{\ensuremath{{#1}_{\textsc{a}}}\xspace}
\newcommand{\ORTRANS}[1]{\ensuremath{{#1}_{\textsc{i}}}\xspace}
\newcommand{\ONLRED}[1]{\ensuremath{(\ORALG{#1},\ORTRANS{#1})}\xspace}
\newcommand{\BLOCKWITHPROBLEM}[4]{\ensuremath{B_{#1}(#2,#3,#4)}\xspace}
\newcommand{\BLOCK}[2]{\ensuremath{B(#1,#2)}\xspace}

\newcommand{\INSTANCES}[1]{\ensuremath{\mathcal{{I}}_{#1}}\xspace}

\makeatletter
\newcommand{\myoverset}[3][0ex]{%
  \mathrel{\mathop{#3}\limits^{
    \vbox to#1{\kern-2\ex@
    \hbox{$\scriptstyle#2$}\vss}}}}
\makeatother
\newcommand{\orarrow}{\ensuremath{\myoverset{r}{\rightarrow}}}

\newcommand{\CCWM}[2]{\ensuremath{\mathcal{C}_{#1}^{#2}}\xspace} 
\newcommand{\CCINFTYWM}[1]{\ensuremath{\mathcal{C}_{#1}}\xspace} 
\newcommand{\ACCWM}[2]{\ensuremath{\mathcal{AC}_{#1}^{#2}}\xspace} 

\newcommand{\OPT}{\ensuremath{\on{\textsc{Opt}}}\xspace}
\newcommand{\ALG}{\ensuremath{\on{\textsc{Alg}}}\xspace}
\newcommand{\FTP}{\ensuremath{\on{\textsc{FtP}}}\xspace}

\newcommand{\LFD}{\ensuremath{\textsc{LFD}}\xspace}
\newcommand{\FLUSH}{\ensuremath{\textsc{FwZ}}\xspace}
\newcommand{\FLUSHwALLzeros}{\ensuremath{\textsc{Flush-when-All-Zeros}}\xspace}
\newcommand{\Block}{\ensuremath{\textsc{FbB}}\xspace}
\newcommand{\FlushBlock}{\ensuremath{\textsc{Flush-between-Blocks}}\xspace}
\newcommand{\subblockOne}{\ensuremath{I_b'}\xspace}
\newcommand{\subblockTwo}{\ensuremath{I_b''}\xspace}

\newcommand{\ASG}[1]{\ensuremath{\textsc{ASG}_{#1}}\xspace}
\newcommand{\ASGINFTY}{\ensuremath{\textsc{ASG}}\xspace}
\newcommand{\BDVC}[1]{\ensuremath{\textsc{VC}_{#1}}\xspace}

\newcommand{\VC}{\ensuremath{\textsc{VC}}\xspace}
\newcommand{\COL}[1]{\ensuremath{#1\textsc{-Spill}}\xspace}
\newcommand{\COLT}[2]{\ensuremath{#1\textsc{-Spill}_{#2}}\xspace}
\newcommand{\INTER}[1]{\ensuremath{\textsc{IR}_{#1}}\xspace}
\newcommand{\IR}{\ensuremath{\textsc{IR}}\xspace}
\newcommand{\SAT}[1]{\ensuremath{#1\textsc{-SatD}}\xspace}
\newcommand{\DOM}{\ensuremath{\textsc{Dom}}\xspace}
\newcommand{\PAG}[1]{\ensuremath{\textsc{Pag}_{#1}}\xspace}

\newcommand{\con}[1]{\ensuremath{\textit{con}(#1)}\xspace}
\newcommand{\dom}[1]{\ensuremath{\textit{dom}(#1)}\xspace}

\newcommand{\ZZ}{\ensuremath{\mathbb{Z}}\xspace}

\newcommand{\RR}{\ensuremath{\mathbb{R}}\xspace}

\newcommand{\ETA}{\ensuremath{\eta}}
\newcommand{\ETAZERO}{\ensuremath{\ETA_0}}
\newcommand{\ETAONE}{\ensuremath{\ETA_1}}
\newcommand{\PAIRETAFORCC}{\ensuremath{\ETAZERO,\ETAONE}}
\newcommand{\PAIRETA}{\ensuremath{(\PAIRETAFORCC)}\xspace}

\newcommand{\XI}{\ensuremath{\xi}}
\newcommand{\XIZERO}{\ensuremath{\XI_0}}
\newcommand{\XIONE}{\ensuremath{\XI_1}}
\newcommand{\PAIRXIFORCC}{\ensuremath{\XIZERO,\XIONE}}
\newcommand{\PAIRXI}{\ensuremath{(\PAIRXIFORCC)}\xspace}

\newcommand{\PHI}{\ensuremath{\varphi}}
\newcommand{\PHIZERO}{\ensuremath{\PHI_0}}
\newcommand{\PHIONE}{\ensuremath{\PHI_1}}
\newcommand{\PAIRPHIFORCC}{\ensuremath{\PHIZERO,\PHIONE}}
\newcommand{\PAIRPHI}{\ensuremath{(\PAIRPHIFORCC)}\xspace}

\newcommand{\MU}{\ensuremath{\mu}}
\newcommand{\MUZERO}{\ensuremath{\MU_0}}
\newcommand{\MUONE}{\ensuremath{\MU_1}}
\newcommand{\PAIRMUFORCC}{\ensuremath{\MUZERO,\MUONE}}
\newcommand{\PAIRMU}{\ensuremath{(\PAIRMUFORCC)}\xspace}

\newcommand{\ZEROMEASURE}{\ensuremath{Z}}
\newcommand{\ZEROMEASUREZERO}{\ensuremath{\ZEROMEASURE_0}}
\newcommand{\ZEROMEASUREONE}{\ensuremath{\ZEROMEASURE_1}}
\newcommand{\PAIRZEROMEASUREFORCC}{\ZEROMEASUREZERO,\ZEROMEASUREONE}

\newcommand{\ALPHA}{\ensuremath{\alpha}}
\newcommand{\BETA}{\ensuremath{\beta}}
\newcommand{\GAMMA}{\ensuremath{\gamma}}
\newcommand{\ABC}{\ensuremath{(\ALPHA,\BETA,\GAMMA)}\xspace}

\tikzset{
	vertex/.style = {
		shape = circle,
		draw = black,
		fill = gray!15,
		minimum size = 1cm,
		scale = 0.8
	}
}

\tikzset{
	vertexwcolor/.style = {
		shape = circle,
		draw = black,
		minimum size = 1cm,
		scale = 0.8
	}
}

\tikzset{
	vertexgray/.style = {
		shape = circle,
		draw = gray,
                fill = gray!5,
		minimum size = 1cm,
		scale = 0.8
	}
}

\tikzset{
	vertexdarkgray/.style = {
		shape = circle,
		draw = gray,
                fill = gray!70,
		minimum size = 1cm,
		scale = 0.8
	}
}

\tikzset{
	smallvertex/.style = {
		shape = circle,
		draw = black,
		fill = gray!15,
		minimum size = 1cm,
		scale = 0.5
	}
}

\tikzset{strictorarrow/.style = {-{Latex[line width=0.5pt, fill=white,scale = 1.5]}}}

\tikzset{orarrow/.style = {-{Latex[scale = 1.5]}}}

\title{Complexity Classes for Online Problems \\
  with and without Predictions\,\thanks{Supported in part by the Independent Research Fund Denmark, Natural Sciences, grants DFF-0135-00018B and DFF-4283-00079B and in part by the Innovation Fund Denmark, grant 9142-00001B, Digital Research Centre Denmark, project P40: Online Algorithms with Predictions. An extended abstract of this paper was published in the International Joint Conference on Theoretical Computer Science -- Frontier of Algorithmic Wisdom (IJTCS-FAW), Track A Best Paper Award, volume 15828 of Lecture Notes in Computer Science, pages 49-63. Springer, 2025.}
}
 
\author{
        \begin{tabular}{l}
          Magnus Berg
        \end{tabular} \\[1ex]
        IBA International Business Academy, Kolding, Denmark \\[1ex]
        \texttt{mbmo@iba.dk}\\[3ex]
        \begin{tabular}{l@{\hspace{2em}}l@{\hspace{2em}}l}
          Joan Boyar & Lene M. Favrholdt & Kim S. Larsen
        \end{tabular} \\[1ex]
        University of Southern Denmark, Odense, Denmark \\[1ex]
        \texttt{\{joan,lenem,kslarsen\}@imada.sdu.dk}
}

\date{June 11, 2026}

\begin{document}

\maketitle

\begin{abstract}
  With the developments in machine learning, there has been a surge in interest and results focused on algorithms utilizing predictions, not least in online algorithms where most new results incorporate the prediction aspect for concrete online problems. While the structural computational hardness of problems with regards to time and space is quite well developed, not much is known about online problems where time and space resources are typically not in focus. Some information-theoretical insights were gained when researchers considered online algorithms with oracle advice, but predictions of uncertain quality is a very different matter.

  We initiate the development of a complexity theory for online problems with predictions, considering minimization problems and one prediction bit per request. Based on the most generic hard online problem type, string guessing, we define a family of hierarchies of complexity classes (indexed by pairs of error measures) and develop notions of reductions, class membership, hardness, and completeness. Our framework contains all the tools one expects to find when working with complexity, and we illustrate our tools by analyzing problems with different characteristics. In addition, we show that known lower bounds for paging with discard predictions apply directly to all hard problems for each class in the hierarchy based on the canonical pair of error measures. This paging problem is not complete for these classes.

  Our work also implies corresponding complexity classes for classic online problems without predictions, with the corresponding complete problems.
\end{abstract}

\section{Introduction}
In computational complexity theory, one aims at classifying
computational problems based on their hardness, by relating them via
hardness-preserving mappings, referred to as reductions. Most
commonly seen are time and space complexity, where problems are
classified based on how much time or space is needed to solve the
problem.  Our primary aim is to classify online minimization problems
with predictions based on the \emph{competitiveness} of
     best possible deterministic online algorithms
for each problem.
In initiating this line of research,
we consider minimization problems with binary predictions,
encoding an optimal solution and given one bit at a time together with the requests.
Our framework has recently been extended to maximization problems~\cite{B25}.

An \emph{online problem} is an optimization problem where the input is
revealed to an online algorithm in a piece-wise fashion in the form of
\emph{requests}. When a request arrives, an online algorithm must make
an irrevocable decision about the request before the next request
arrives.  When comparing the quality of online algorithms, we use the
standard \emph{competitive analysis} framework~\cite{ST85} (see~\cite{BE98,K16}), where
the competitiveness of an online algorithm is computed by comparing
the algorithm's performance
to an offline optimal
algorithm's performance.
Competitive analysis is
a framework for worst-case guarantees, where we say that an algorithm
is $c$-competitive if, asymptotically over all possible input
sequences, its cost is at most a factor $c$ times the cost of the
optimal offline algorithm.

With the increased availability and improved quality of predictions
from machine learning software, efforts to utilize predictions in
online algorithms have increased dramatically~\cite{ALPS}. Typically,
one studies the competitiveness of online algorithms that have access
to additional information about the instance through (unreliable)
predictions.  Ideally, such algorithms should perform perfectly when
the predictions are error-free (the competitiveness in this case is
called the \emph{consistency}), and
perform as well as the best purely online algorithm when the
predictions are erroneous (\emph{robustness}).  There is also a desire
that an algorithm's competitiveness degrades gracefully from the
consistency to the robustness as the predictions get worse (often
referred to as \emph{smoothness}). In particular, the performance should not
plummet due to minor errors. To establish smoothness, it is necessary to
have some measure of how wrong a prediction is. Thus, results of this
type are based on some \emph{error measure}.

The complexity of algorithms with predictions has also been considered in a
different context, dynamic graph problems~\cite{HSSY24}.
However,
Henzinger et al.\ study the \emph{time} complexity of dynamic data structures,
whereas we create complexity classes where the hardness is based on \emph{competitiveness}.

The basis for our complexity classes is a parameterized version of
\emph{asymmetric string
guessing}~\cite{BFKM17}, a generic hard
online problem,
where each request
is simply a prompt for the
algorithm to guess a bit.
String guessing (not necessarily asymmetric)~\cite{BHKKSS14} has played a
fundamental role in what is often referred to as advice
complexity~\cite{BKKKM17,DKP09,EFKR11,HKK10,BFKLM17j}, where online algorithms have access to oracle-produced information about the instance which, in our context, can be considered infallible predictions.
The same standard string guessing problem has also been used for the
advice complexity of priority algorithms~\cite{BBLP20,BLP24j,BNR03}.
Specifically, we use \emph{Online $(1,t)$-Asymmetric
String Guessing with Unknown History and Predictions} ($\ASG{t}$), which
will be our base family of complete problems, establishing a strict
hierarchy based on the parameter,~$t$. The \emph{cost} of processing
an input
is the
number of guesses of $1$
plus $t$ times
the number of incorrect guesses of $0$.
Other variants of
Asymmetric String Guessing
have been used before in~\cite{M16}, considering connections
between advice complexity and randomization, as well as in~\cite{BFKM18} where
weighted versions of problems proven hard with respect to advice complexity
in~\cite{BFKM17} are considered.

We define complexity classes,
$\CCWM{\PAIRETAFORCC}{t}$, parameterized by
$t\in \ZZ^+ \cup \{\infty\}$ and a pair of error measures, \PAIRETA,
with certain properties.  To prove that a problem, $P$, is
$\CCWM{\PAIRETAFORCC}{t}$-\emph{hard}, one must 
show that $P$ is as hard as $\ASG{t}$,
and to prove \emph{membership} in $\CCWM{\PAIRETAFORCC}{t}$, one
must show
that $\ASG{t}$ is as hard as $P$.
If both are true, $P$ is
$\CCWM{\PAIRETAFORCC}{t}$-\emph{complete}.  
The as-hard-as relation is
transitive, so our framework provides all the usual tools: if a
subproblem of some problem is hard, the problem itself is hard,
one can reduce from the most convenient complete problem,
etc.
Thus, working with our complexity classes is similar to working
with, e.g., NP, MAX-SNP~\cite{PY91}, the
W-hierarchy~\cite{DF99}, and APX~\cite{APMGCK99}, in that hardness
results are obtained by proving the existence of special types of
reductions that preserve properties related to hardness.
However, we obtain performance bounds that are independent of any conjectures.

Deriving lower bounds on the competitiveness
of algorithms based on the
hardness of string guessing has been considered
before~\cite{BHKKSS14,BBLP20,EFKR11},
with different
objectives. The closest related work is in \cite{BFKM17}, where one of the base problems
we use in this paper,  $(1,\infty)$-Asymmetric String Guessing with Unknown History, was used as the base problem for the complexity
class AOC; AOC-complete problems are hard online problems
with advice.
Note that despite the similarities, working with advice and predictions are quite different matters.
In advice complexity, the competitive ratio is a function of the number of advice bits available.
Working with predictions, competitiveness is a function of the quality, not the quantity, of information about the input. 
Thus, results cannot be translated between AOC and the complexity classes of this paper.
Moreover, AOC is only one complexity class, not a hierarchy.

Strong lower
bounds in the form of hardness results from our framework can be seen
as indicating the insufficiency of a binary prediction scheme for a
problem. Proving that a problem is $\CCWM{\PAIRETAFORCC}{t}$-hard
suggests that one cannot solve it better than blindly trusting the
predictions, when using binary predictions, giving a rather poor result.
Hence, proving that a problem is hard serves as an argument for
needing a richer prediction scheme for the problem, or possibly a
more accurate way of measuring prediction error.

Our main contribution is an initial framework enabling a complexity theory for online algorithms with binary predictions. In this framework,
proving that a problem is hard for a class currently requires that the predictions
be binary encodings of the output the online algorithm should produce.
Using this framework, we prove hardness and class membership results for
several problems, including showing the completeness of Online $t$-Bounded
Degree Vertex Cover ($\BDVC{t}$) for $\CCWM{\PAIRETAFORCC}{t}$. 
Thus, $\BDVC{t}$, or any other complete problem, could be used as the basis for the complexity classes instead of $\ASG{t}$. However, we follow the tradition from advice complexity and use a string guessing problem, $\ASG{t}$, as its lack of structure
offers
simpler proofs. 
We illustrate
the relative hardness of the problems we investigate in
Figure~\ref{fig:hardness_graph}.
\begin{figure}[tb]
\centering
\begin{tikzpicture}[scale = 0.9]
\node (asgtm1) at (1.5,0) {$\ASG{t-1}$};
\node (asgt) at (5,0) {$\ASG{t}$};
\draw[strictorarrow] (asgtm1) -- (asgt);
\node (asgtp1) at (8.5,0) {$\ASG{t+1}$};
\draw[strictorarrow] (asgt) -- (asgtp1);

\node at (-0.5,0) {$\cdots$};
\node at (11.5,0) {$\cdots$};
\draw[strictorarrow] (0,0) -- (asgtm1);
\draw[strictorarrow] (asgtp1) -- (11,0);

\node (vctm1) at (2.5,1.5) {$\BDVC{t-1}$};
\node (vct) at (6,1.5) {$\BDVC{t}$};
\node (vctp1) at (9.5,1.5) {$\BDVC{t+1}$};

\draw[orarrow] (asgt) to[out = 45, in = -110] (vct);
\draw[orarrow] (vct) to[out = 225, in = 70] (asgt);

\draw[orarrow] (asgtm1) to[out = 45, in = -110] (vctm1);
\draw[orarrow] (vctm1) to[out = 180+45, in = 180-110] (asgtm1);

\draw[orarrow] (asgtp1) to[out = 45, in = -110] (vctp1);
\draw[orarrow] (vctp1) to[out = 180+45, in = 180-110] (asgtp1);

\node (istm1) at (0.5,1.5) {$\INTER{t-1}$};
\node (ist) at (4,1.5) {$\INTER{t}$};
\node (istp1) at (7.5,1.5) {$\INTER{t+1}$};

\draw[orarrow] (asgtm1) to[out = 110, in = -45] (istm1);
\draw[orarrow] (istm1) to[out = -70, in = 135] (asgtm1);
\draw[orarrow] (asgt) to[out = 110, in = -45] (ist);
\draw[orarrow] (ist) to[out = -70, in = 135] (asgt);
\draw[orarrow] (asgtp1) to[out = 110, in = -45] (istp1);
\draw[orarrow] (istp1) to[out = -70, in = 135] (asgtp1);
\draw[orarrow] (istm1) to[out = 10, in = 180-10] (vctm1);
\draw[orarrow] (vctm1) to[out = 180+10, in = -10] (istm1);
\draw[orarrow] (ist) to[out = 10, in =170] (vct);
\draw[orarrow] (vct) to[out = 180+10, in = -10] (ist);
\draw[orarrow] (istp1) to[out = 10, in =170] (vctp1);
\draw[orarrow] (vctp1) to[out = 180+10, in = -10] (istp1);

\node (vc) at (6.75,3) {$\VC$};
\node (dom) at (6.75,4.5) {$\DOM$};

\draw[strictorarrow] (vctm1) -- (vc);
\draw[strictorarrow] (vct) -- (vc);
\draw[strictorarrow] (vctp1) -- (vc);
\draw[orarrow,dashed] (vc) -- (dom);

\node (inter) at (3.25,3) {$\IR$};
\draw[strictorarrow] (istm1) -- (inter);
\draw[strictorarrow] (ist) -- (inter);
\draw[strictorarrow] (istp1) -- (inter);
\draw[orarrow] (inter) -- (vc);

\node (2sat) at (3.25,4.5) {$\SAT{2}$};
\draw[orarrow] (inter) -- (2sat);

\node (asg) at (5,3.75) {$\ASGINFTY$};
\draw[strictorarrow] (vc) -- (asg);
\draw[strictorarrow] (dom) -- (asg);

\node (kcolm) at (-0.7,-1.5) {$\COLT{k}{t+k}$};
\draw[orarrow] (asgtm1) -- (kcolm);
\node (kcol) at (4,-1.5) {$\COLT{k}{t+k+1}$};
\draw[orarrow] (asgt) -- (kcol);
\node (kcolp) at (8.5,-1.5) {$\COLT{k}{t+k+2}$};
\draw[orarrow] (asgtp1) -- (kcolp);

\node (pagtm1) at (1.5,-1.5) {$\PAG{t-1}$};
\draw[strictorarrow,gray!75] (pagtm1) -- (asgtm1);
\node (pagt) at (6,-1.5) {$\PAG{t}$};
\draw[strictorarrow,gray!75] (pagt) -- (asgt);
\node (pagtp1) at (11,-1.5) {$\PAG{t+1}$};
\draw[strictorarrow,gray!75] (pagtp1) -- (asgtp1);

\end{tikzpicture}
\caption{A hardness graph based on our complexity hierarchy.
  The problems shown are defined in Definitions~\ref{def:asg_t},~\ref{def:bdvc_t},~\ref{def:ir_t},~\ref{def:k-spill},~\ref{def:2-SAT},~\ref{def:dom}, and~\ref{def:paging}.
Given two problems $P$ and $Q$, we write $P \rightarrow Q$ to indicate that $Q$ is as hard as $P$ (see Definition~\ref{def:as_hard_as}).
If the arrowhead is only outlined, $P$ is not as hard as $Q$.
If the arrow is dashed, $P$ is asymptotically as hard as $Q$ (see Definition~\ref{def:weaklyhard}).
We leave out most arrows that can be derived by transitivity.
The \textcolor{gray!75}{gray} arrows
hold with respect to the pair of error measures $\PAIRMU$ (see Definition~\ref{def:error_measures_for_ASG_analysis}), and the remaining arrows hold with respect to all pairs of \emph{insertion monotone} error measures (see Definition~\ref{def:insertion_monotone_error_measures}).
}
\label{fig:hardness_graph}
\end{figure}
Worth noting is that by choosing the appropriate pair of error
measures, our set-up immediately gives the same hardness results for purely
online problems, that is, for algorithms without predictions.

As a soundness test for our framework, we consider the paging problem.
Intuitively, researchers in online algorithms would expect that problem
to be easier than, for instance, $\ASG{t}$ or $\BDVC{t}$,
because it seems that much more information is available.
And, indeed, we can prove that, as opposed to those two problems,
paging with discard predictions is not a complete problem
for the complexity class $\CCWM{\PAIRMUFORCC}{t}$ from our framework.
Furthermore, maybe surprisingly, this
non-complete paging problem gives rise to new lower bounds for
$\CCWM{\PAIRMUFORCC}{t}$-complete problems.

\section{Preliminaries}\label{sec:prelims}

In this paper, we consider online problems with binary predictions. From now on, we will not keep emphasizing that they are binary but simply refer to them as \emph{predictions}. The algorithms have to make an irrevocable decision, $y_i$, for each request. With the exception of the paging problem, we focus on problems, where these decisions (the $y_i$'s) are also binary.

For any problem, $P$, as just described, we let $\INSTANCES{P}$ be the collection of instances of $P$, we let $\ALGS{P}$ be the set of algorithms for $P$, and we let $\OPT_P$ be a fixed optimal algorithm for $P$.
In our notation, an instance of $P$ is a triple $I = (x,\hat{x},r)$, consisting of two bitstrings $x,\hat{x} \in \{0,1\}^n$, and a sequence of requests $r = \langle r_1,r_2,\ldots,r_n\rangle$ for $P$.
The bitstring $x$ is an encoding of $\OPT_P$'s solution, and $\hat{x}$ is a prediction of $x$.
When an algorithm, $\ALG$, receives the request $r_i$, it also receives the prediction $\hat{x}_i$ to aid its decision (represented by a bit, $y_i$) for $r_i$.
What information is contained in each request, $r_i$, and the meaning of the bits $x_i$, $\hat{x}_i$, and $y_i$, will be specified for each problem.
When there can be no confusion, we write $\OPT$ instead of $\OPT_P$.

Given an algorithm, $\ALG \in \ALGS{P}$, and an instance, $I \in \INSTANCES{P}$, we let $\ALG[I]$ be $\ALG$'s solution to $I$, and $\ALG(I)$ be the cost of $\ALG[I]$.

\subsection{Competitiveness and Error Measures}
For purely online algorithms, we use the following definition of competitiveness:
An algorithm, $\ALG$, for a minimization problem, is \emph{$c$-competitive} if there exists a constant, $\AT$, called the \emph{additive term}, such that for all instances $I = (x,r) \in \INSTANCES{P}$, 
\begin{align*}
\ALG(I) \leqslant c \cdot \OPT(I) + \AT.
\end{align*}

We extend the definition of competitiveness to online algorithms with predictions, based on~\cite{ABEFHLPS23}.
Here, the competitiveness of an algorithm is written as a function of two error measures\footnote{We work with separate error measures for
predicted $0$s and $1$s to allow for more detailed results. Our
reductions and structural results would also work if we used only one
error measure. For instance, letting $\eta_{t-1,1}(x,\hat{x}) =
(t-1) \cdot \MUZERO(x,\hat{x}) + \MUONE(x,\hat{x})$,
Theorem~\ref{thm:ftp_equality_wrt_mu_0_and_mu_1} implies that $\FTP$
is $(1,1)$-competitive with respect to $\eta_{t-1,1}$.
However,
combining the two error measures into one, we lose some detail such as the trade-offs between $\ALPHA$, $\BETA$, and $\GAMMA$ given in
Theorem~\ref{thm:stronger_lower_bounds_from_paging_general_version}.}, $\ETAZERO$ and $\ETAONE$, where $\ETA_b$ is a function of the bits incorrectly predicted to be $b$.
\begin{definition}\label{def:competitiveness}
Let $\PAIRETA$ be a pair of error measures, let $P$ be an online minimization problem with predictions, and let $\ALG$ be a deterministic online algorithm for $P$.
If there exist three maps $\ALPHA, \BETA, \GAMMA \colon \INSTANCES{P} \rightarrow \RR_{\geqslant 0}$ and an additive constant, \AT, such that for all $I \in \INSTANCES{P}$,
\begin{align*}
\ALG(I) \leqslant \ALPHA \cdot \OPT(I) + \BETA \cdot \ETAZERO(I) + \GAMMA \cdot \ETAONE(I) + \AT,
\end{align*}
then $\ALG$ is \emph{$\ABC$-competitive with respect to $\PAIRETA$}.
When $\PAIRETA$ is clear from the context, we simply write that $\ALG$ is $\ABC$-competitive.
If $\AT \leqslant 0$, for all $I \in \INSTANCES{P}$, then $\ALG$ is \emph{strictly $\ABC$-competitive}. 
If \ALG is not $(\alpha,\beta,\gamma)$-competitive for any $\alpha,\beta,\gamma$, we simply say that it is {\em not competitive}.
\end{definition}
In our notation, $\ALPHA$, $\BETA$, and $\GAMMA$'s dependency on $I$ is kept implicit, as is tradition.
For the results in this paper, it is not necessary that $\alpha$, $\beta$, and $\gamma$ are functions of the input.
However, we are developing a framework that should be generally applicable, and many results from standard competitive analysis have competitive ratios that are functions of the input (size). Examples include graph coloring, dynamic binary search trees, and some scheduling problems.

Further, observe that any purely online $\ALPHA$-competitive algorithm is $(\ALPHA,0,0)$-competitive with respect to any pair of error measures $\PAIRETA$.

Since the performance of an algorithm is not measured by a competitive ratio, but by a triple, $(\alpha,\beta,\gamma)$, we cannot establish a total ordering of algorithms. Instead, we consider the concept of Pareto-optimality.

\begin{definition}
An $\ABC$-competitive algorithm is called \emph{Pareto-optimal} for a problem, $P$, if, for any $\eps > 0$, there cannot exist an $(\alpha - \eps ,\beta,\gamma)$-, an $(\alpha,\beta-\eps,\gamma)$-, or an $(\alpha,\beta,\gamma - \eps)$-competitive algorithm for $P$.
\end{definition}

\subsubsection{Error Measures of Interest}

We define a pair of error measures, $(\mu_0,\mu_1)$, equivalent to the pair of error measures from~\cite{ABEFHLPS23}, where $\MU_b$ is the number of incorrect predictions of $b \in \{0,1\}$:

\begin{definition}\label{def:error_measures_for_ASG_analysis}
For any instance, $I = (x,\hat{x},r)$, 
\begin{align*}
&\MUZERO(I) = \sum_{i=1}^n x_i \cdot (1-\hat{x}_i) \; \text{ and } \;
 \MUONE(I) = \sum_{i=1}^n (1-x_i) \cdot \hat{x}_i.
\end{align*}
\end{definition}

The following pair of error measures allows us to extend the results of this paper to purely online algorithms.
\begin{definition}\label{def:error_measure_Z}
For any instance, $I = (x,\hat{x},r)$,
$\ZEROMEASUREZERO(I) = \ZEROMEASUREONE(I) = 0$.
\end{definition}

In this paper, we will consider error measures with the property that
insertion of a finite number of correctly predicted bits into an instance does not increase
the error:
\begin{definition}\label{def:insertion_monotone_error_measures}
  An error measure $\ETA$ is called \emph{insertion monotone} if for any instance $I$,
$\ETA(I') \leqslant \ETA(I)$,
where $I'$ is obtained by inserting a finite number of correctly predicted requests into $I$.
\end{definition}

Clearly $\MU_b$ and $\ZEROMEASURE_b$, $b \in \{0,1\}$, are insertion monotone.\footnote{Other examples of insertion monotone functions are $L^p$ norms. Moreover, any linear combination of insertion monotone functions is insertion monotone.}
Although formally, the predictions $\hat{x}$ exist in general results that
assume insertion monotone measures,
these predictions are not used in the analysis when $\ZEROMEASURE_b$ is used;
the competitiveness must be obtained independently of the correctness of the prediction.
Thus, these
results hold for purely online algorithms, where the predictions do not exist.

\subsection{Graph Problems}
All graph problems mentioned in this paper are studied in the \emph{vertex-arrival} model, the most standard for online graph problems. 
Hence, for any graph $G = (V,E)$, each request $r_i$ is a vertex $v_i \in V$ that is revealed together with all edges of the form $(v_j,v_i) \in E$, where $j \leqslant i$.

We only consider simple unweighted graphs.

\section{Asymmetric String Guessing: \\ A Collection of Hard Problems}

Given a bitstring, $x$, and a $t \in \ZZ^+ \cup \{\infty\}$, the task of
an algorithm for $(1,t)$-Asymmetric String Guessing (\ASG{t}) is to correctly guess the contents of $x$.
The cost of a solution is the number of guesses of $1$ plus $t$ times
the number of incorrect guesses of $0$.
When $t=\infty$, the problem corresponds to the string guessing
problem considered in~\cite{BFKM17}.
We now define the problem more formally.

\begin{definition}\label{def:asg_t}
  For any $t \in \ZZ^+ \cup \{\infty\}$, an instance of the problem \emph{Online $(1,t)$-Asymmetric
  String Guessing with Unknown History and Predictions} ($\ASG{t}$) is
  a triple $I = (x,\hat{x},r)$, where $x = \langle x_1, \ldots x_n
  \rangle$ and $\hat{x} = \langle \hat{x}_1, \ldots, \hat{x}_n \rangle$ are
  bitstrings and $r=\langle r_1, \ldots, r_n \rangle$ is a
  sequence of requests.
  Each request, $r_i$, is a prompt for the algorithm to output a bit, $y_i$.
  Together with $r_i$, $\hat{x}_i$ is revealed, but $x$ is
  only revealed after the last request. 

Given an instance $I \in \INSTANCES{\ASG{t}}$ of $\ASG{t}$ with $t \in \ZZ^+$,
\begin{align*}
\ALG(I) = \sum_{i=1}^n \left(y_i + t \cdot x_i\cdot(1-y_i) \right),
\end{align*}
where $y_i$ is $\ALG$'s $i$'th guess.

When $t = \infty$, we abbreviate $\ASG{t}$ by $\ASGINFTY$ and rewrite the objective function as:
\begin{align*}
\ALG(I) = \begin{cases}
\displaystyle \sum_{i=1}^n y_i, &\text{ if } \: \displaystyle \sum_{i=1}^n x_i\cdot(1-y_i) = 0 \\
\infty, &\mbox{otherwise.}
\end{cases}
\end{align*}
\end{definition}

\subsection{$\boldsymbol{\ASG{t}}$ without Predictions}
For all $t \in \ZZ^+ \cup \{\infty\}$, we may consider $\ASG{t}$ as a purely online problem by omitting $\hat{x}$.
We briefly state the main results on $\ASG{t}$ without predictions.
The following observation is well-known.

\begin{observation}(\cite{BFKM17})\label{obs:no_competitive_algorithm_for_asg_infty}
  For any purely online algorithm, $\ALG$, for $\ASGINFTY$, there is no function, $f$, such that $\ALG$ is $f(n)$-competitive.
\end{observation}

The observation follows from the fact that if an algorithm, $\ALG$, ever guesses $0$, there is an instance where \ALG guesses $0$ on a true $1$, and thus incurs cost $\infty$.
On the other hand, if $\ALG$ only guesses $1$, there is an instance consisting of only $0$'s, such that $\OPT$ will incur cost $0$ while the cost of \ALG is equal to the length of the sequence. 

\begin{theorem}\label{thm:always_guess_zero}
Let $t \in \ZZ^+$ and $\eps > 0$.
Then, for $\ASG{t}$, the following hold.
\begin{enumerate}[label = {(\roman*)}]
\item The algorithm that always guesses $0$ is $t$-competitive. \label{item:always_guess_0}
\item There is no purely online $(t-\eps)$-competitive deterministic algorithm. \label{item:no_t-eps_com_det_alg}
\end{enumerate}
\end{theorem}
\begin{proof}
  We consider each item separately.
  
\textbf{Towards~\ref{item:always_guess_0}:}
Let $\ALG$ be the algorithm that always guesses $0$. 
Then, for any instance $I = (x,r) \in \INSTANCES{\ASG{t}}$, 
\begin{align*}
\ALG(I) = \sum_{i=1}^{n} t \cdot x_i = t \cdot \sum_{i=1}^{n} x_i = t \cdot \OPT(I).
\end{align*}

\textbf{Towards~\ref{item:no_t-eps_com_det_alg}:}
Assume towards contradiction that there exists a deterministic purely online $(t-\eps)$-competitive algorithm, $\ALG$, for some $\eps>0$.
Then, there exists a constant $0<\eps'<1$ such that $\ALG$ is $(t-\eps')$-competitive.
This means that there exists an additive term, $\AT \in O(1)$, such that, for all $I = (x,r) \in \INSTANCES{\ASG{t}}$,
\begin{align}\label{eq:c_comp_alg_no_preds}
\ALG(I) \leqslant (t-\eps') \cdot \OPT(I) + \AT.
\end{align}

We define a family $\{I^n\}_{n\in\ZZ^+}$ of instances, where, for any $n \in \ZZ^+$ and any $i \in \{1,2,\ldots,n\}$, the $i$'th true bit in $I^n = (x^n,r^n)$ is 
\begin{align*}
x^n_i = \begin{cases}
0, &\mbox{if $y_i^n = 1$}, \\
1, &\mbox{if $y_i^n = 0$},
\end{cases}
\end{align*}
where $y_i^n$ is $\ALG$'s $i$'th guess when run on $I^n$.
Since $\ALG$ is deterministic, the collection $\{I^n\}_{n\in\ZZ^+}$ is well-defined.

For each $i = 1,2,\ldots,n$, if $y^n_i = 0$, then $x_i^n = 1$, and so $\ALG$ incurs cost $t$, and $\OPT$ incurs cost $1$. 
On the other hand, if $y_i^n = 1$, then $x_i^n = 0$, and so $\ALG$ incurs cost $1$, and $\OPT$ incurs cost $0$.
Hence, for each $n \in \ZZ^+$,
\begin{align}\label{eq:alg_cost}
  \ALG(I^n)
  & =  t \cdot \OPT(I^n) + 1 \cdot (n - \OPT(I^n)) \nonumber\\
  & = (t-1) \cdot \OPT(I^n) + n.
\end{align}
Combining Equation~\eqref{eq:alg_cost} with Inequality~\eqref{eq:c_comp_alg_no_preds}, we get
\begin{align*}
(t-1) \cdot \OPT(I^n) + n \leqslant (t-\eps') \cdot \OPT(I^n) + \AT.
\end{align*}
Solving for $\AT$, we get
\begin{align}\label{eq:kappa_lower}
  \AT
  &\geqslant n - (1-\eps') \cdot \OPT(I^n) \nonumber \\
  &\geqslant \eps' \cdot n, \text{ since } 1-\eps'>0 \text{ and } \OPT \leqslant n.
\end{align}
This contradicts that $\AT \in O(1)$.
\end{proof}

\subsection{$\boldsymbol{\ASG{t}}$ with Predictions and Error Measures $\boldsymbol{(\mu_0,\mu_1)}$}

Now, we turn to the hardness of $\ASG{t}$ with predictions, focusing on the pair $(\mu_0,\mu_1)$ of error measures.

An obvious algorithm for $\ASG{t}$ is \emph{Follow-the-Predictions}~($\FTP$), which always sets its guess, $y_i$, to the given prediction, $\hat{x}_i$.

\begin{theorem}\label{thm:ftp_equality_wrt_mu_0_and_mu_1}
  For any $t \in \ZZ^+$ and any $\alpha$ and $\beta$ such that $\alpha \geqslant 1$ and $\alpha + \beta \geqslant t$,
  $\FTP$ is strictly $(\alpha,\beta,1)$-competitive for $\ASG{t}$ with respect
  to $\PAIRMU$.
  \end{theorem}
\begin{proof}
  Consider any $I = (x,\hat{x},r) \in \INSTANCES{\ASG{t}}$ with $n = \abs{r}$, and let $y_i$ be \FTP's output on $r_i$, $1 \leqslant i \leqslant n$.
  Then,
\begin{align*}
\FTP(I) =\; &\sum_{i=1}^n \left( y_i + t \cdot x_i \cdot (1-y_i) \right) \\
=\; &\sum_{i=1}^n \left( \hat{x}_i + t \cdot x_i \cdot (1-\hat{x}_i) \right) \\
\leqslant\; &\sum_{i=1}^n \hat{x}_i + (\alpha+\beta) \cdot \sum_{i=1}^n \left(x_i \cdot (1-\hat{x}_i) \right), \text{ since } \alpha + \beta \geqslant t \\
=\; &\alpha \cdot \sum_{i=1}^n x_i + \beta \cdot \sum_{i=1}^n \left(x_i \cdot (1-\hat{x}_i) \right) + \sum_{i=1}^n (1-\alpha \cdot x_i) \cdot \hat{x}_i \\
\leqslant\; &\alpha \cdot \sum_{i=1}^n x_i + \beta \cdot \sum_{i=1}^n \left(x_i \cdot (1-\hat{x}_i) \right) + \sum_{i=1}^n (1- x_i) \cdot \hat{x}_i, \text{ since } \alpha \geqslant 1 \\
=\; &\alpha \cdot \OPT(I) + \beta \cdot \mu_0(I) + \mu_1(I).
\end{align*}
\end{proof}

\begin{corollary}\label{cor:FTP}
  For any $t \in \ZZ^+$,
  $\FTP$ is strictly $(1,t-1,1)$-competitive for $\ASG{t}$ with respect
  to $\PAIRMU$.
\end{corollary}

In the following, we extend a lower bound on Paging with Discard Predictions by Antoniadis et al.~\cite{ABEFHLPS23} to $\ASG{t}$.
For completeness, we state the theorem here, though restricted to $\ASG{t}$ and without proof.
In Theorem~\ref{thm:stronger_lower_bounds_from_paging_general_version} in Section~\ref{sec:pagingBounds}, it is restated and proven in a more general form.

\begin{theorem}\label{thm:stronger_lower_bounds_from_paging}
Let $t \in \ZZ^+$.
Then, for any $\ABC$-competitive algorithm for $\ASG{t}$ with respect to $\PAIRMU$,
\begin{enumerate}[label = {(\roman*)}]
\item $\ALPHA + \BETA \geqslant t$ and \label{thm:paging-1}
\item $\ALPHA + (t-1) \cdot \GAMMA \geqslant t$. \label{thm:paging-2}
\end{enumerate}
\end{theorem}

Next, we prove a negative result on the competitiveness of algorithms for $\ASG{t}$ that does not follow from
Theorem~\ref{thm:stronger_lower_bounds_from_paging}.
Together with Theorems~\ref{thm:always_guess_zero},~\ref{thm:ftp_equality_wrt_mu_0_and_mu_1}, and~\ref{thm:stronger_lower_bounds_from_paging}, this negative
result gives a complete classification of all Pareto-optimal algorithms for $\ASG{t}$ with respect to $\PAIRMU$.

\begin{lemma}\label{lem:small_alpha_gives_large_gamma}
Let $\ALG$ be an $\ABC$-competitive algorithm for $\ASG{t}$ with respect to $\PAIRMU$.
If $\alpha < t$, then $\gamma \geqslant 1$.
\end{lemma}
\begin{proof}
  Assume towards contradiction that $\ALG$ is $\ABC$-competitive, where
  $\alpha<t$ and $\gamma<1$.
  Then there exists an $\eps > 0$ such that $\alpha \leqslant t-\eps$ and $\gamma \leqslant 1-\eps$.

  Consider the family $\{I^n\}_{n\in\ZZ^+}$ of instances with $I^n=(x^n,\hat{x}^n,r^n)$, where
  \begin{itemize}
  \item $\hat{x}^n = \langle 1^n \rangle$ and
  \item $x^n_i = 1 - y^n_i$, where $y^n_i$ is \ALG's output on $r^n_i$, $1 \leqslant i \leqslant n$.
  \end{itemize}
Observe that $x^n$ is well-defined since $\ALG$ is deterministic.

Then,
\begin{align}
  \ALG(I^n)
  & = \sum_{i=1}^n \left( y^n_i + t \cdot (1-y^n_i) \right) \nonumber\\
  & = \sum_{i=1}^n \left( 1-x^n_i + t \cdot x^n_i \right), \text{ by definition of $x^n$} \nonumber\\
  & = n + (t-1) \cdot \sum_{i=1}^n x^n_i \label{eq:alg-1}
\end{align}

Since $\ALG$ is $(t-\eps,\beta,1-\eps)$-competitive, there exists a $\AT \in O(1)$, such that, for each $n$,
\begin{align}
\ALG(I^n)
  &\leqslant (t-\eps) \cdot \OPT(I^n) + (1-\eps) \cdot \mu_1(I^n) + \AT, \text{ since } \mu_0(I^n)=0  \nonumber\\
  &= (t-\eps) \cdot \sum_{i=1}^n x^n_i + (1-\eps) \cdot \sum_{i=1}^n (1-x^n_i)  + \AT\nonumber\\
  &= (t-1) \cdot \sum_{i=1}^n x^n_i + (1-\eps) \cdot n + \AT \label{eq:alg-2}
\end{align}

Combining~(\ref{eq:alg-1}) and~(\ref{eq:alg-2}) and solving for $\AT$, we get
\begin{align*}
 & n + (t-1) \cdot \sum_{i=1}^n x^n_i \leqslant (t-1) \cdot \sum_{i=1}^n x^n_i + (1-\eps) \cdot n + \AT \; \Leftrightarrow\\
 & \eps \cdot n < \AT,
\end{align*}
contradicting that $\AT \in O(1)$, since $\eps>0$.
\end{proof}

\begin{theorem}\label{thm:Pareto}
Let $\ALG$ be an $\ABC$-competitive algorithm for $\ASG{t}$.
Then, $\ALG$ is Pareto-optimal with respect to $(\mu_0,\mu_1)$, if and only if
\begin{enumerate}[label = {(\roman*)}]
\item $\alpha = t$ and $\beta = \gamma = 0$, or \label{item:pareto_optimal_purely_online}
\item $\alpha < t$, $\beta = t - \alpha$, and $\gamma = 1$.\label{item:pareto_optimal_predictions}
\end{enumerate}
\end{theorem}
\begin{proof}
We consider an $(\alpha,\beta,\gamma)$-competitive algorithm and split the proof into two cases based on the value of $\alpha$.

\textbf{Case~$\alpha \geqslant t$:}
By Theorem~\ref{thm:always_guess_zero}\ref{item:always_guess_0}, there exists a $(t,0,0)$-competitive algorithm.
This algorithm is Pareto-optimal, since neither the second nor the third entry can be improved and by Theorem~\ref{thm:always_guess_zero}\ref{item:no_t-eps_com_det_alg}, no $(t-\eps,0,0)$-competitive algorithm exists.
Thus, since $\alpha \geqslant t$, \ALG is Pareto-optimal, if and only if $\alpha=t$ and $\beta=\gamma=0$.

\textbf{Case~$\alpha < t$:}
We note that $\alpha \geqslant 1$, since no algorithm can be better than $1$-consistent.
Thus, by Theorem~\ref{thm:ftp_equality_wrt_mu_0_and_mu_1}, the algorithm \FTP is $(\alpha, t-\alpha, 1)$-competitive.
Hence, for \ALG to be Pareto-optimal, we must have that
\begin{enumerate}[label = {(\alph*)}]
\item $\beta < t-\alpha$, \label{item:pareto_alpha-beta}
\item $\gamma < 1$, or \label{item:pareto_gamma}
\item $\beta=t-\alpha$ and $\gamma=1$.
\end{enumerate}
By Theorem~\ref{thm:stronger_lower_bounds_from_paging}\ref{thm:paging-1}, \ref{item:pareto_alpha-beta} is impossible, and, by Lemma~\ref{lem:small_alpha_gives_large_gamma}, \ref{item:pareto_gamma} is impossible.
Thus, \ALG is Pareto-optimal, if and only if $\beta=t-\alpha$ and $\gamma=1$.
\end{proof}

Note that by Theorem~\ref{thm:always_guess_zero}\ref{item:always_guess_0}, the algorithm that always guesses zero is $t$-robust with respect to \PAIRMU.

\begin{corollary}
\FTP  and the algorithm that always guesses~0 are both Pareto-optimal with respect to $(\mu_0,\mu_1)$.
\end{corollary}

\begin{theorem}
  There is no competitive algorithm 
  for $\ASGINFTY$ with respect to $\PAIRMU$.
\end{theorem}
\begin{proof}
  Suppose towards contradiction that
  for all sequences $I$, $\ALG(I) \leqslant \alpha\OPT(I)+\beta\MUZERO(I)
  +\gamma\MUONE(I)+\AT$, where
  $\ALPHA, \BETA, \GAMMA \colon \INSTANCES{P} \rightarrow \RR_{\geqslant 0}$
  and $\AT$ is a constant.
Consider the family $\{I^n\}_{n\in\ZZ^+}$ of instances with $I^n=(x^n,\hat{x}^n,r^n)$, where
  \begin{itemize}
  \item $\hat{x}^n = \langle 0^n \rangle$ and
  \item $x^n_i = 1 - y^n_i$, where $y^n_i$ is \ALG's output on $r^n_i$, $1 \leqslant i \leqslant n$.
  \end{itemize}
Observe that $x^n$ is well-defined since \ALG is deterministic.
Also, note that
$\MUONE(I^n)=0$, so $\ALG(I^n) \leqslant \alpha\OPT(I^n)+\beta\MUZERO(I^n) +\AT$.

If $y^n_i=0$ for some $i$, then $x^n_i=1$, so \ALG incurs an infinite cost on $r_i$.
Thus, since neither $\alpha$ nor $\beta$ can be infinite, we have a contradiction.

Otherwise, $y^n_i=1$ for all $i$, so $\ALG(I^n)=n$, $\OPT(I^n)=0$ and
$\MUZERO(I^n) = 0$, implying that
$n\leqslant \alpha\cdot 0+\beta\cdot 0+\AT = \AT\in O(1)$,
giving a contradiction.

Thus, \ALG is not competitive.
\end{proof}

\section{Hierarchies of Complexity Classes}

In this section, we formally introduce the complexity classes, prove that, for each pair of error measures, they form a strict hierarchy, and show multiple fundamental structural properties of the complexity classes.

\subsection{Relative Hardness and Reductions}

We define relative hardness as follows:

\begin{definition}\label{def:as_hard_as}
Let $P$ and $Q$ be two online problems with predictions and error measures $\PAIRETA$ and \PAIRPHI.
We say that $Q$ is \emph{as hard as} $P$ \emph{with respect to \PAIRPHI and \PAIRETA}, if the existence of an $\ABC$-competitive algorithm for $Q$ with respect to $\PAIRPHI$ implies the existence of an $\ABC$-competitive algorithm for $P$ with respect to $\PAIRETA$.
If the error measures are clear from the context, we simply say that $Q$ is \emph{as hard as} $P$.
If $Q$ is as hard as $P$, we use the notation $Q\geq_{\textrm o} P$.
We also say that $P$ is \emph{no harder than} $Q$, denoting
this $P \leq_{\textrm o} Q$.
\end{definition}

It is not hard to see that the as-hard-as relation is both reflexive and transitive, but for completeness we give the proof here.

\begin{lemma}\label{lem:transitivity}
The relation as-hard-as is reflexive and transitive. 
\end{lemma}
\begin{proof}
We prove each property separately.

\textbf{Towards reflexivity:}
Proving reflexivity translates to proving that the existence of an $\ABC$-competitive algorithm with respect to $\PAIRETA$ for $P$ implies the existence of an $\ABC$-competitive algorithm with respect to $\PAIRETA$ for $P$, which is a tautology.

\textbf{Towards transitivity:}
Let $P$, $W$, and $Q$ be online maximization problems with predictions with error measures $\PAIRETA$, $\PAIRXI$, and $\PAIRPHI$, and assume that $Q\geq_{\textrm o} W$ and that $W \geq_{\textrm o} P$.
Let $\ALG_Q \in \ALGS{Q}$ be an $\ABC$-competitive algorithm for $Q$ with respect to $\PAIRPHI$.
Since $Q \geq_{\textrm o} W$ with respect to $\PAIRPHI$ and $\PAIRXI$, the existence of $\ALG_Q$ implies the existence of an $\ABC$-competitive algorithm with respect to $\PAIRXI$ for $W$, say $\ALG_W$. 
Since $W \geq_{\textrm o} P$ with respect to $\PAIRXI$ and $\PAIRETA$, the existence of $\ALG_W$ also implies the existence of an $\ABC$-competitive algorithm for $P$ with respect to $\PAIRETA$, and so $Q \geq_{\textrm o} P$ with respect to $\PAIRPHI$ and $\PAIRETA$.
\end{proof} 

As a tool for proving hardness, we introduce the notion of reductions.
A reduction from a problem, $P$, to another problem, $Q$, consists of
a mapping of instances of $P$ to instances of $Q$ and a mapping from
algorithms for $Q$ to algorithms for $P$, with the requirement that
$(\alpha, \beta, \gamma)$-competitive algorithms for $Q$ map to
$(\alpha, \beta, \gamma)$-competitive algorithms for $P$.
In this paper, we use a restricted type of reduction:

\begin{definition}\label{def:OR}
  Let $P$ and $Q$ be online minimization problems with predictions, and let $\PAIRETA$ and $\PAIRPHI$ be pairs of error measures for the predictions in $P$ and $Q$, respectively.
  Let $\rho = \ONLRED{\rho}$ be a tuple consisting of two maps, $\ORALG{\rho}\colon\ALGS{Q}\rightarrow\ALGS{P}$ and $\ORTRANS{\rho}\colon\ALGS{Q} \times \INSTANCES{P} \rightarrow \INSTANCES{Q}$.

  If there exists a constant, $a$, called the \emph{reduction term} of $\rho$, such that for each instance $I_P \in \INSTANCES{P}$ and each algorithm $\ALG_Q \in \ALGS{Q}$, letting $\ALG_P = \ORALG{\rho}(\ALG_Q)$ and $I_Q = \ORTRANS{\rho}(\ALG_Q,I_P)$,
\begin{enumerate}[label = {(O\arabic*)}]
\item $\ALG_P(I_P) \leqslant \ALG_Q(I_Q) + a$,\label{item:OR_condition_ALG}
\item $\OPT_Q(I_Q) \leqslant \OPT_P(I_P)$, \label{item:OR_condition_OPT}
\item $\PHIZERO(I_Q) \leqslant \ETAZERO(I_P)$, and $\PHIONE(I_Q) \leqslant \ETAONE(I_P)$, \label{item:OR_condition_eta}
\end{enumerate}
then $\rho$ is called a \emph{strict online reduction from $P$ to $Q$ with respect to $\PAIRETA$ and $\PAIRPHI$}.
If
the pairs $\PAIRETA$ and $\PAIRPHI$ are clear from the context, then we simply say that $\rho$ is
a \emph{strict online reduction} from $P$ to $Q$, and write $\rho \colon P \orarrow Q$.
\end{definition}

In many of the examples to follow, we do not need a (nonzero) reduction term to establish the reductions, and we only mention the reduction term when needed.

\begin{lemma}\label{lem:why_the_naming}
  Let $\rho = \ONLRED{\rho}$ be a strict online reduction from a problem, $P$, with predictions and error measures \PAIRETA to a problem, $Q$, with predictions and error measures \PAIRPHI.
  Let $\ALG_Q \in \ALGS{Q}$ be an $\ABC$-competitive algorithm for $Q$ with respect to $\PAIRPHI$, and let $\ALG_P = \ORALG{\rho}(\ALG_Q)$.
  Then, $\ALG_P$ is an $\ABC$-competitive algorithm for $P$ with respect to $\PAIRETA$.
\end{lemma}
\begin{proof}
Consider any instance, $I_P \in \INSTANCES{P}$, and let $I_Q = \ORTRANS{\rho}(\ALG_Q,I_P)$.
Then, by Condition~\ref{item:OR_condition_ALG} and Definition~\ref{def:competitiveness}, there exists constants, $a$ and \AT, such that
\begin{align*}
\ALG_P(I_P) \leqslant & \ALG_Q(I_Q) + a \\ 
\leqslant & \: \ALPHA \cdot \OPT_Q(I_Q) + \BETA \cdot \PHIZERO(I_Q) + \GAMMA \cdot \PHIONE(I_Q) + \AT + a \\
\leqslant & \: \ALPHA \cdot \OPT_P(I_P) + \BETA \cdot \ETAZERO(I_P) + \GAMMA \cdot \ETAONE(I_P) + \AT + a, \\
          & \text{ by~\ref{item:OR_condition_OPT} and~\ref{item:OR_condition_eta}} 
\end{align*}
\end{proof}

\begin{observation}\label{obs:reductions_hardness}
  Lemma~\ref{lem:why_the_naming} implies that strict online
reductions serve the desired purpose of reductions: If there
exists a strict online reduction $\rho \colon P \orarrow Q$,
then $Q \geq_{\textrm o} P$.
\end{observation}
When using strict online reductions, we will often simply use the term
\emph{reduction}.
  
For the rest of this paper, we
only
consider reductions
where the quality of predictions
is measured using the same pair of error measures for both problems.

\subsection{Introducing the Complexity Classes}\label{sec:complexity_classes_def}

For any pair, $\PAIRETA$, of error measures and any $t \in \ZZ^+ \cup \{\infty\}$, we define the complexity classes $\CCWM{\PAIRETAFORCC}{t}$ as the set of minimization problems with predictions that are no harder than $\ASG{t}$ with respect to $\PAIRETA$:

\begin{definition}\label{def:complexityClasses}
For each $t \in \ZZ^+ \cup\{\infty\}$ and each pair of error measures, $\PAIRETA$, the complexity class $\CCWM{\PAIRETAFORCC}{t}$ is the closure of $\ASG{t}$ under the as-hard-as relation with respect to $\PAIRETA$. 
Hence, for an online minimization problem, $P$, 
\begin{itemize}
\item $P \in \CCWM{\PAIRETAFORCC}{t}$, if $\ASG{t} \geq_{\textrm o} P$,
\item  $P$ is \emph{$\CCWM{\PAIRETAFORCC}{t}$-hard}, if $P \geq_{\textrm o} \ASG{t}$, and
\item $P$ is \emph{$\CCWM{\PAIRETAFORCC}{t}$-complete}, if $P \in \CCWM{\PAIRETAFORCC}{t}$ and $P$ is $\CCWM{\PAIRETAFORCC}{t}$-hard.
\end{itemize}
The notation \CCWM{\PAIRETAFORCC}{\infty} is usually abbreviated \CCINFTYWM{\PAIRETAFORCC}.
\end{definition}

\begin{observation}\label{obs:ASG_classes} By Observation~\ref{obs:reductions_hardness} and
  Lemma~\ref{lem:why_the_naming}, if $P$ and $Q$ are $\CCWM{\PAIRETAFORCC}{t}$-complete problems, then there exists an $\ABC$-competitive algorithm for $P$ if and only if there exists an $\ABC$-competitive algorithm for $Q$. Thus, all of the results for $\ASG{t}$ hold for the complete problems in $\CCWM{\PAIRETAFORCC}{t}$, and the upper bounds for $\ASG{t}$ hold for all of the problems in $\CCWM{\PAIRETAFORCC}{t}$.
\end{observation}

Since the as-hard-as relation is reflexive, $\ASG{t}$ is $\CCWM{\PAIRETAFORCC}{t}$-complete for any pair of error measures, $\PAIRETA$, and any $t$. 
Further, due to transitivity, we have the following simple but important properties that one would expect should hold for complexity classes.

\begin{theorem}\label{lem:different_base_problem}
Let $t \in \ZZ^+ \cup \{\infty\}$ and let $\PAIRETA$ be any pair of error measures.
\begin{enumerate}[label = {(\roman*)}]
\item If $P \in \CCWM{\PAIRETAFORCC}{t}$ and $P \geq_{\textrm o} Q$, then $Q \in \CCWM{\PAIRETAFORCC}{t}$. \label{item:sub_problems}
\item If $P$ is $\CCWM{\PAIRETAFORCC}{t}$-hard and $Q \geq_{\textrm o} P$, then $Q$ is $\CCWM{\PAIRETAFORCC}{t}$-hard.  \label{item:sup_problems}
\end{enumerate}
\end{theorem}
\begin{proof}
We prove each item separately.
  
\textbf{Towards~\ref{item:sub_problems}:}
Since $P \in \CCWM{\PAIRETAFORCC}{t}$, $\ASG{t} \geq_{\textrm o} P$.
Since $P \geq_{\textrm o} Q$, transitivity implies that $\ASG{t} \geq_{\textrm o} Q$, and thus $Q \in \CCWM{\PAIRETAFORCC}{t}$.

\textbf{Towards~\ref{item:sup_problems}:}
Since $P$ is $\CCWM{\PAIRETAFORCC}{t}$-hard, $P \geq_{\textrm o} \ASG{t}$.
Since $Q \geq_{\textrm o} P$, transitivity implies that $Q \geq_{\textrm o} \ASG{t}$, and thus we conclude that $Q$ is $\CCWM{\PAIRETAFORCC}{t}$-hard.
\end{proof}

This theorem implies results concerning special cases of a problem:

\begin{corollary}\label{cor:sub_and_sup_problem}
Let $t \in \ZZ^+ \cup \{\infty\}$ and let $\PAIRETA$ be any pair of error measures.
Let $P$ and $P_{\on{sub}}$ be online minimization problems such that $\INSTANCES{P_{\on{sub}}}
\subseteq \INSTANCES{P}$.
\begin{enumerate}[label=(\roman*)]
\item If $P \in \CCWM{\PAIRETAFORCC}{t}$, then $P_{\on{sub}} \in \CCWM{\PAIRETAFORCC}{t}$. \label{item:sub}
\item If $P_{\on{sub}}$ is $\CCWM{\PAIRETAFORCC}{t}$-hard, then $P$ is $\CCWM{\PAIRETAFORCC}{t}$-hard. \label{item:sup}
\end{enumerate}
\end{corollary}
\begin{proof}
There exists a trivial online reduction, $\rho \colon P_{\on{sub}} \orarrow P$, obtained by setting $\ORALG{\rho}(\ALG) = \ALG$ and $\ORTRANS{\rho}(\ALG,I) = I$, for all algorithms $\ALG \in \ALGS{P}$ and all instances $I \in \INSTANCES{P_{\on{sub}}}$.
Hence, this is a consequence of Theorem~\ref{lem:different_base_problem}.
\end{proof}

Note that Theorem~\ref{lem:different_base_problem} implies the following.
\begin{observation}\label{obs:anyothercanbeused}
  Any $\CCWM{\PAIRETAFORCC}{t}$-complete problem can be used instead of \ASG{t} as the base problem when defining $\CCWM{\PAIRETAFORCC}{t}$.
\end{observation}
As shown in Lemmas~\ref{lem:bdvc_t_hardness} and~\ref{lem:bdvc_t_containment}, the better known Online $t$-Bounded Degree Vertex Cover with Predictions is $\CCWM{\PAIRETAFORCC}{t}$-complete with respect to a wide range of pairs of error measures and may therefore be used as the basis of these complexity classes instead of $\ASG{t}$. 
However, we chose to define the complexity classes as the closure of $\ASG{t}$, due to it being a generic problem that is easily analyzed.
Also, after establishing the first $\CCWM{\PAIRETAFORCC}{t}$-hard problems, we may reduce from any $\CCWM{\PAIRETAFORCC}{t}$-hard problem to prove hardness.
Finally, a problem $Q$ is $\CCWM{\PAIRETAFORCC}{t}$-hard, if and only if $Q \geq_{\textrm o} P$, for all $P \in \CCWM{\PAIRETAFORCC}{t}$.
This is in line with the structure of other complexity classes such as NP and APX, where a problem $Q$ is NP-hard, respectively APX-hard, if and only if there exists a polynomial-time reduction, respectively PTAS-reduction, from any problem in NP, respectively APX, to $Q$.

\begin{proposition}\label{prop:ASG_complete}
  For every problem $P$ in $\CCWM{\PAIRMUFORCC}{t}$, there exists a $(t,0,0)$-com\-peti\-tive algorithm and for every
  $1 \leqslant \alpha \leqslant t$, there exists a $(\alpha,t-\alpha,1)$-competitive algorithm
  for $P$. If $P$ is $\CCWM{\PAIRMUFORCC}{t}$-complete, this is tight for Pareto-optimal algorithms for $P$.
\end{proposition}
\begin{proof}
  By Theorems~\ref{thm:ftp_equality_wrt_mu_0_and_mu_1} and~\ref{thm:Pareto}, the theorem holds for $\ASG{t}$.
  Thus, by Observation~\ref{obs:ASG_classes}, it also holds for any problems in $\CCWM{\PAIRETAFORCC}{t}$. 
  \end{proof}

Note that by Proposition~\ref{prop:ASG_complete}, all of the problems in $\CCWM{\PAIRMUFORCC}{t}$  have a $(1,t-1,1)$-competitive algorithm, so there are algorithms with consistency~1.
Moreover, problems in the class indexed by $t$ have a $t$-robust algorithm, and for complete problems this is tight.

In introducing some initial complexity classes to the field of online algorithms (with and without predictions), we define reductions from complete problems for each class to all members of the class. The reductions map algorithms for one problem to algorithms for another problem. To show inclusion of a problem $P$ in a class, the map is from $P$ to a complete problem.
The original complete problems for each class have binary outputs and predictions.
In contrast, the Paging problem, with discard predictions for whether or not the current page will be in a selected \OPT's cache the next time it is requested, obtains an optimal solution by choosing pages to evict among those that would not be in \OPT's cache. Thus, the predictions are an encoding of all optimal solutions that would have a page in cache, if the selected \OPT would. The sequence of pages chosen for eviction is the output.

\subsection{Establishing the Hierarchy}

In this subsection, we show that our complexity classes form a strict
hierarchy by showing that $\ASG{t+1}$ is strictly harder than $\ASG{t}$, for any $t \in \ZZ^+$.

\begin{lemma}\label{lem:hierarchy_lemma}
For any $t \in \ZZ^+$ and any pair of error measures, $\PAIRETA$, 
\begin{enumerate}[label = {(\roman*)}]
\item $\ASG{t+1} \geq_{\textrm o} \ASG{t}$, and \label{item:or_from_asg_t_to_asg_t+1}
\item $\ASG{t}$ is not as hard as $\ASG{t+1}$. \label{item:no_or_from_asg_t+1_to_asg_t}
\end{enumerate}
\end{lemma}
\begin{proof}
We prove each item separately.
  
\textbf{Towards~\ref{item:or_from_asg_t_to_asg_t+1}:}
We give a reduction, $\rho = \ONLRED{\rho}$, from $\ASG{t}$ to $\ASG{t+1}$.
For any  $I=(x,\hat{x},r) \in \INSTANCES{\ASG{t}}$ and $\ALG_{t+1} \in \ALGS{\ASG{t+1}}$, we let $\ORTRANS{\rho}(\ALG_{t+1}, I) = I$.
We let $\ALG_{t} = \ORALG{\rho}(\ALG_{t+1})$ be the algorithm which, for each request, passes the prediction bit on to $\ALG_{t+1}$ and outputs the same bit as $\ALG_{t+1}$.

We verify that this reduction satisfies Conditions~\ref{item:OR_condition_ALG}--\ref{item:OR_condition_eta} from Definition~\ref{def:OR}.
Since $\ORTRANS{\rho}(\ALG_{t+1},I) = I$, Conditions~\ref{item:OR_condition_OPT}--\ref{item:OR_condition_eta} are immediate for any pair of error measures $\PAIRETA$. 
Towards Condition~\ref{item:OR_condition_ALG}, denote by $y_1,y_2,\ldots,y_n$ the bits guessed by $\ALG_{t+1}$, and therefore also by $\ALG_t$.
Then,
\begin{align*}
\ALG_{t+1}(I) - \ALG_t(I) =& \sum_{i=1}^n \left(y_i + (t+1) \cdot (1-y_i) \cdot x_i\right) \\
&- \sum_{i=1}^n \left(y_i + t\cdot (1-y_i) \cdot x_i\right) \\
=& \sum_{i=1}^n (1-y_i) \cdot x_i \geqslant 0,
\end{align*}
and Condition~\ref{item:OR_condition_ALG} follows.

\textbf{Towards~\ref{item:no_or_from_asg_t+1_to_asg_t}:}
Assume towards contradiction that $\ASG{t} \geq_{\textrm o} \ASG{t+1}$.
Then, for any $\ABC$-competitive algorithm for $\ASG{t}$, there exists an $\ABC$-competitive algorithm for $\ASG{t+1}$.
By Theorem~\ref{thm:always_guess_zero}\ref{item:always_guess_0}, the algorithm that always guesses $0$ is $(t,0,0)$-competitive for $\ASG{t}$.
Hence, there exists a $(t,0,0)$-competitive algorithm for $\ASG{t+1}$, which contradicts Theorem~\ref{thm:always_guess_zero}\ref{item:no_t-eps_com_det_alg}.
\end{proof}

\begin{lemma}
For any $t \in \ZZ^+$ and any pair of error measures, $\PAIRETA$, $\CCWM{\PAIRETAFORCC}{t} \subsetneq \CCWM{\PAIRETAFORCC}{t+1}$. 
\end{lemma}
\begin{proof}
We first prove that $\CCWM{\PAIRETAFORCC}{t} \subseteq \CCWM{\PAIRETAFORCC}{t+1}$. 
For any $P \in \CCWM{\PAIRETAFORCC}{t}$, $\ASG{t} \geq_{\textrm o} P$, and, by Lemma~\ref{lem:hierarchy_lemma}\ref{item:or_from_asg_t_to_asg_t+1}, $\ASG{t+1} \geq_{\textrm o} \ASG{t}$.
Thus, by transitivity (Lemma~\ref{lem:transitivity}), $\ASG{t+1} \geq_{\textrm o} P$, and so $P \in \CCWM{\PAIRETAFORCC}{t+1}$.
To see that $\CCWM{\PAIRETAFORCC}{t} \subsetneq \CCWM{\PAIRETAFORCC}{t+1}$, observe that $\ASG{t+1} \in \CCWM{\PAIRETAFORCC}{t+1}$ and, by Lemma~\ref{lem:hierarchy_lemma}\ref{item:no_or_from_asg_t+1_to_asg_t}, $\ASG{t+1} \not\in \CCWM{\PAIRETAFORCC}{t}$.
\end{proof}

\begin{lemma}
For any $t \in \ZZ^+$ and any pair of error measures, $\PAIRETA$, no
$\CCWM{\PAIRETAFORCC}{t+1}$-hard problem is in $\CCWM{\PAIRETAFORCC}{t}$.
\end{lemma}
\begin{proof}
Assume towards contradiction that $\CCWM{\PAIRETAFORCC}{t}$ contains a $\CCWM{\PAIRETAFORCC}{t+1}$-hard problem, $P$.
Then, $\ASG{t} \geq_{\textrm o} P$, and $P \geq_{\textrm o} \ASG{t+1}$.
Thus, due to transitivity, $\ASG{t} \geq_{\textrm o} \ASG{t+1}$, contradicting Lemma~\ref{lem:hierarchy_lemma}.
\end{proof}

By similar arguments, we get the following:
\begin{lemma}
For any $t \in \ZZ^+$ and any pair of error measures, $\PAIRETA$, 
\begin{itemize}
\item $\CCWM{\PAIRETAFORCC}{t} \subsetneq \CCINFTYWM{\PAIRETAFORCC}$, and
\item $\CCWM{\PAIRETAFORCC}{t}$ contains no $\CCINFTYWM{\PAIRETAFORCC}$-hard problems.
\end{itemize}
\end{lemma}

\begin{theorem}\label{thm:hierarchy}
For any pair of error measures $\PAIRETA$, we have a strict hierarchy of complexity classes:
\begin{align*}
\CCWM{\PAIRETAFORCC}{1} \subsetneq \CCWM{\PAIRETAFORCC}{2} \subsetneq \CCWM{\PAIRETAFORCC}{3} \subsetneq \cdots \subsetneq \CCINFTYWM{\PAIRETAFORCC}.
\end{align*}
\end{theorem}

\subsection{Purely Online Algorithms}

Observe that our complexity theory extends to a complexity theory for purely online algorithms as well. 
In particular, one may consider the complexity classes, $\CCWM{\PAIRZEROMEASUREFORCC}{t}$.
Recall that $\ZEROMEASUREZERO(I) = \ZEROMEASUREONE(I) = 0$, for any instance $I = (x,\hat{x},r)$.
In this framework, any $\ABC$-competitive algorithm, $\ALG$, for an online minimization problem $P$, satisfies that
\begin{align*}
\ALG(I) &\leqslant \ALPHA \cdot \OPT(I) + \BETA \cdot \ZEROMEASUREZERO(I) + \GAMMA \cdot \ZEROMEASUREONE(I) + \AT \\
&= \ALPHA \cdot \OPT(I) + \AT,
\end{align*}
for all instances $I \in \INSTANCES{P}$, and so $\ALG$ is an $\ALPHA$-competitive purely online algorithm for $P$. Since \ASG{t} has a $t$-competitive
algorithm (it guesses all $0$s) and is as hard as all problems in $\CCWM{\PAIRZEROMEASUREFORCC}{t}$, all problems in $\CCWM{\PAIRZEROMEASUREFORCC}{t}$
are (at worst) $t$-competitive.

By the above and Theorem~\ref{lem:different_base_problem}, if a purely online problem, $P$, has an $\alpha$-competitive algorithm, it is in
$\CCWM{\PAIRZEROMEASUREFORCC}{t}$, for any $t \geqslant \alpha$.
If, in addition, $P$ has no $(t-\eps)$-competitive algorithm, $P$ is
$\CCWM{\PAIRZEROMEASUREFORCC}{t}$-complete.

Since most results in this paper hold for any insertion monotone error measure, and $\ZEROMEASUREZERO$ and $\ZEROMEASUREONE$ are insertion monotone, 
we obtain
a similar complexity theory for purely online algorithms.

\begin{observation}
  All results in this paper on problems and complexity classes with respect to insertion monotone error measures also hold for the same problems and classes without predictions.
\end{observation}

Note that our complexity theory may help translating results for one purely online problem directly to other purely online problems.
In Section~\ref{sec:pagingBounds}, we discuss a strategy for proving lower bounds for all $\CCWM{\PAIRMUFORCC}{t}$-hard problems.
Using the same strategy, the lower bound from Theorem~\ref{thm:always_guess_zero}\ref{item:no_t-eps_com_det_alg} on the competitive ratio of any purely online algorithm for $\ASG{t}$ extends to a lower bound on the competitiveness of any purely online algorithm for any $\CCWM{\PAIRZEROMEASUREFORCC}{t}$-hard problems.

\begin{observation}
  Assume that some minimization problem, $Q$, is as hard as some minimization problem, $P$, with respect to any pair of insertion monotone error measures.
  Then any lower bound on the competitiveness of the purely online version of $P$ is a lower bound on the competitiveness of the purely online version of $Q$.
  Similarly, any upper bound for the purely online version of $Q$ is an upper bound for the purely online version of $P$.
\end{observation}

\section{A Template for Online Reductions via Simulation}\label{sec:reduction_template}

In this section, we introduce a template for creating online reductions from $\ASG{t}$ to a problem, $Q$, implying $\CCWM{\PAIRETAFORCC}{t}$-hardness of $Q$ with respect to any pair of 
insertion monotone error measures.
Recall from Definition~\ref{def:insertion_monotone_error_measures} that an error measure is insertion monotone if
insertion of a finite number of correctly predicted bits into an instance does not increase
the error. 
The reduction strategy is outlined in Algorithm~\ref{alg:redalg}.
It might be an advantage to read this template in conjunction with the first example of its use for Vertex Cover (Lemma~\ref{lem:bdvc_t_hardness} in Section~\ref{sec:VC}) to better understand how elements from the concrete problems enter into the template.

\begin{algorithm}[h!]
\caption{Reducing \ASG{t} to $Q$}
\begin{algorithmic}[1]
\Require $\ALG' \in \ALGS{Q}$ and $I = (x,\hat{x},r) \in \INSTANCES{\ASG{t}}$
\Ensure $\ALG \in \ALGS{\ASG{t}}$ and $I' = (x',\hat{x}',r') \in \INSTANCES{Q}$ \Comment{$I'$ may be longer than $I$}
\For {each request, $r_i$, in $r$}
	\State Give $\ALG'$ the \emph{challenge request} $c_i$ and the prediction $\hat{x}_i' = \hat{x}_i$ \label{template:cp}
	\State Output $y_i = y_i'$, where $y_i'$ is the output of $\ALG'$ for $c_i$ \label{template:y}
\EndFor
\State Receive $x = x_1,x_2,\ldots,x_n$ \Comment{Finishes the $\ASG{t}$ instance}  \label{template:x}
\For {$j=1,2,\ldots,n$} \label{template:blockstart}
        \For {each request, $r$, in \emph{block} $\BLOCKWITHPROBLEM{Q}{x_j}{y_j'}{j}$}
             \State Give $\ALG'$ the request $r$ together with a correct prediction of $0$ \label{template:blockend}
        \EndFor
\EndFor
\end{algorithmic}
\label{alg:redalg}
\end{algorithm}
Algorithm~\ref{alg:redalg} is a template where the challenge requests and the blocks need to be defined when a concrete online reduction is created.
For each request, $r_i$, of the $\ASG{t}$ instance, we create and give a challenge request, $c_i$, to the algorithm, $\ALG'$, for $Q$ in a way such that
no information about the true bit, $x_i'$, can be inferred from the information about the instance obtained so far.
If, for instance, $Q$ is $\BDVC{t}$, each challenge request will be a vertex which is isolated at arrival.
After we have treated the entire $\ASG{t}$ instance, we give $\ALG'$ one (possibly empty)
block of requests per challenge request.
For $\BDVC{t}$, the $i$th block consists of a (possibly empty) set of neighbors of the $i$th challenge request, which an optimal algorithm can reject.
The purpose of the blocks is to ensure that
$x'_i=x_i$, for $1 \leqslant i \leqslant n$, and $x'_i=0$, for $i>n$, encode an optimal solution for $Q$.
Thus, for each incorrect bit, $\hat{x}_i'$, in $\hat{x}'$, the bit $\hat{x}_i$ is also incorrect, though $\hat{x}'$ has additional, correct, predictions associated with the blocks. Since the error measures, $\eta_0$ and $\eta_1$, are insertion monotone, $\eta_0(I')\leqslant \eta_0(I)$ and $\eta_1(I')\leqslant \eta_1(I)$.

Summing up, in the input created for $\ALG'$, the requests are the challenge requests given in Line~\ref{template:cp} followed by the block requests (Lines~\ref{template:blockstart}--\ref{template:blockend}), the prediction consists of the prediction bits for the $\ASG{t}$ instance (Line~\ref{template:cp}) followed by a 0 for each block request, and the optimal solution is the correct solution to the $\ASG{t}$ instance (Line~\ref{template:x}) followed by a 0 for each block request.
On each request, $r_i$, \ALG outputs the same as $\ALG'$ does on $c_i$.

Note that $\ALG'$ may not recognize when the challenge requests end
and the blocks begin, in which case it may not answer all of the
challenges in the blocks correctly, despite the correct predictions.
However, this is not a problem, since the idea is to prove that $\ALG$ does no worse than $\ALG'$.

In the online reductions, we need to know which request we are dealing with to create the right blocks. However, most often this is clear from the context, and then we leave out the third parameter to $\BLOCKWITHPROBLEM{Q}{x_j}{y_j'}{j}$.

\section{Examples of $\boldsymbol{\CCWM{\PAIRETAFORCC}{t}}$-Hard Problems}

We use the reduction template from Section~\ref{sec:reduction_template} for proving $\CCWM{\PAIRETAFORCC}{t}$-hardness of Vertex Cover (Lemma~\ref{lem:bdvc_t_hardness}), $k$-Spill (Theorem~\ref{thm:k-col_hardness}), and Interval Rejection (Lemma~\ref{lem:interval_scheduling_hardness}). In addition, we demonstrate the use of transitivity to
prove hardness, creating reductions from Interval Rejection to $2$-\textsc{SAT}-Deletion (Theorem~\ref{thm:2_sat_hardness}) and from  Vertex Cover to Dominating Set (Theorem~\ref{thm:hardness_dominating_set}).

\subsection{Vertex Cover}\label{sec:VC}

Given a graph, $G = (V,E)$, an algorithm for Vertex Cover finds a  subset $V' \subseteq V$ of vertices such that for each edge $e = (u,v) \in E$, $u \in V'$ or $v \in V'$.
The cost of the solution is given by the size of $V'$, and the goal is to minimize this cost.
In $t$-Bounded Degree Vertex Cover, all vertices have degree at most~$t$.

Other work on Vertex Cover includes the following.
In the purely online setting,
there exists a $t$-competitive algorithm for $t$-Bounded Degree Vertex Cover, and the problem does not allow for a $(t-\eps)$-competitive algorithm, for any $\eps > 0$~\cite{DP05}.
Considering advice, the problem is AOC-complete~\cite{BFKM17}.
In the offline setting, Vertex Cover is NP-complete~\cite{K72} and APX-complete~\cite{S82,DS05}, and
$t$-Bounded Degree Vertex Cover is  MAX-SNP-hard~\cite{PY91}.

We formally define the Vertex Cover problem studied in this section:

\begin{definition}\label{def:bdvc_t}
The input to \emph{Online $t$-Bounded Degree Vertex Cover with Predictions} ($\BDVC{t}$) is a graph, $G = (V,E)$, with maximum degree at most $t$.
When an algorithm, $\ALG$, receives a vertex, $v_i$, it outputs $y_i = 1$ to
accept this vertex or $y_i = 0$ to reject it.

For the algorithm's solution to be feasible, it must be a vertex cover, i.e., $$y_i=1 \vee y_j=1, \text{ for each edge } (v_i,v_j) \in E.$$
For an instance, $I = (x,\hat{x},r) \in \INSTANCES{\BDVC{t}}$,
$$\ALG(I) = \sum_{i=1}^n y_i.$$
The set $\{v_i \mid x_i=1\}$ is an optimal vertex cover.
\end{definition}

The standard unbounded Online Vertex Cover with Predictions is also considered, and is abbreviated $\VC$.

First, we show hardness results for $\BDVC{t}$:
\begin{lemma}\label{lem:bdvc_t_hardness}
For any $t \in \ZZ^+$ and any pair of insertion monotone error measures, $\PAIRETA$, $\BDVC{t}$ is $\CCWM{\PAIRETAFORCC}{t}$-hard.
\end{lemma}
\begin{proof}
  We give a strict online reduction $\rho = \ONLRED{\rho}$ from \ASG{t} to \BDVC{t}, using the reduction template of Algorithm~\ref{alg:redalg} in Section~\ref{sec:reduction_template}
(see Figure~\ref{fig:bounded_degree_vertex_cover_example} for an example).
Consider any $I = (x,\hat{x},r) \in \INSTANCES{\ASG{t}}$ and any $\ALG' \in \ALGS{\BDVC{t}}$.
Let $\ALG = \ORALG{\rho}(\ALG')$ and $I'=\ORTRANS{\rho}(\ALG',I)$ be the \ASG{t} algorithm and \BDVC{t} instance created in the reduction described below.

\begin{figure}[htp]
\centering
\begin{tikzpicture}
\node at (-1.5,5.2) {$x_i=x'_i$:};  
\node at (-1.5,3.8) {$y_i=y'_i$:};  
\begin{scope}[shift={(1.5,0)}]
  \node[vertex] (v2) at (0,4.5) {$v_2$};
  \node at (0.15,5.2) {$0$};
  \node at (0.15,3.8) {$1$};
\end{scope}
  
\begin{scope}[shift={(3,0)}]
  \node[vertex] (v4) at (1.5,4.5) {$v_4$};
  \node at (1.65,5.2) {$1$};
  \node at (1.65,3.8) {$0$};
  \node[vertex] (v21) at (1.5,3) {$v_{4,1}$};
  \draw (v21) -- (v4);
  \node[vertex] (v22) at (1.5,1.5) {$v_{4,2}$};
  \draw (v22) to[out = 125, in = 235] (v4);
  \node[vertex] (v23) at (1.5,0) {$v_{4,3}$};
  \draw (v23) to[out = 135, in = 225] (v4);
\end{scope}

\begin{scope}[shift={(-3,0)}]
  \node[vertex] (v1) at (3,4.5) {$v_1$};
  \node at (3.15,5.2) {$0$};
  \node at (3.15,3.8) {$0$};
\end{scope}

\begin{scope}[shift={(-1.5,0)}]
  \node[vertex] (v3) at (4.5,4.5) {$v_3$};
  \node at (4.65,5.2) {$1$};
  \node at (4.65,3.8) {$1$};
  \node[vertex] (v71) at (4.5,3) {$v_{3,1}$};
  \draw (v71) -- (v3);
\end{scope}

\begin{scope}[shift={(-0.15,-8.5)}]
  \node at (-2,7.4) {$\OPT = \OPT'$:}; \node at (0.15,7.4) {$0$}; \node at (1.65,7.4) {$0$}; \node at (3.15,7.4) {$1$};           \node at (4.65,7.4) {$1$}; 
  \node at (-2,6.8) {\ALG:};           \node at (0.15,6.8) {$0$}; \node at (1.65,6.8) {$1$}; \node at (3.15,6.8) {$1$};           \node at (4.65,6.8) {$t$}; 
  \node at (-2,6.2) {$\ALG'$:};        \node at (0.15,6.2) {$0$}; \node at (1.65,6.2) {$1$}; \node at (3.15,6.2) {$\geqslant 1$}; \node at (4.65,6.2) {$t$};
\end{scope}
\end{tikzpicture}
\caption{A \BDVC{3} instance constructed from an \ASG{3} instance as described in the proof of Lemma~\ref{lem:bdvc_t_hardness}.
  For each challenge request, the bits $x_i$ and $y_i$ are shown above and below $v_i$, respectively.
Below the graph, the cost of each challenge request and its corresponding block requests is shown for $\OPT$, $\OPT'$, $\ALG$, and $\ALG'$.}
\label{fig:bounded_degree_vertex_cover_example}
\end{figure}

The $i$th challenge request is an isolated vertex, $v_i$.
Clearly, $v_1, v_2, \ldots, v_i$ give no information about $x_i$,
as required.
Recall from the template that $\hat{x}_i' = \hat{x}_i$ and that
\ALG outputs $y_i=y'_i$, where $y_i'$ is the output of $\ALG'$ for $v_i$.

Following all the $n$ challenge requests, we give the $n$ blocks, where
the $i$th block, $\BLOCK{x_i}{y_i'}$, is constructed as follows, with all prediction bits equal to $0$ as described in the template.
\begin{itemize}
\item If $x_i = 0$, then $\BLOCK{x_i}{y_i'}$ is empty. Thus, no optimal solution will contain $v_i$.
\item If $x_i = y_i' = 1$, then $\BLOCK{x_i}{y_i'}$ contains one request to a vertex, $v_{i,1}$, connected to $v_i$, ensuring that there is an optimal solution containing $v_i$.
\item If $x_i = 1$ and $y_i' = 0$, then $\BLOCK{x_i}{y_i'}$ contains requests to $t$ new vertices, $v_{i,j}$, $j=1,2,\ldots,t$, each connected to $v_i$, giving $\ALG'$ a cost of $t$ and ensuring that there is an optimal solution containing $v_i$. 
\end{itemize}

Observe that $\{ v_i \mid x_i=1 \}$ is an optimal solution to the \BDVC{t} instance constructed, as required by the template.
This shows that $\OPT(I)=\OPT(I')$, satisfying Condition~\ref{item:OR_condition_OPT} of Definition~\ref{def:OR}.
It also shows that all block requests are correctly predicted,
since all block predictions are zero.
Since $\hat{x}'_i=\hat{x}_i$, for each challenge request, $v_i$, and $\eta_0$ and $\eta_1$ are insertion monotone, we conclude that $\eta_0(I') \leqslant \eta_0(I)$ and $\eta_1(I') \leqslant \eta_1(I)$, satisfying Condition~\ref{item:OR_condition_eta} of Definition~\ref{def:OR}.

To prove that Condition~\ref{item:OR_condition_ALG} is also satisfied, we argue that $\ALG(I)\leqslant\ALG'(I')$.
For each request, $r_i$ in $r$, such that $x_i = 0$, $\ALG$ and $\ALG'$ both have a cost of $y_i = y_i'$ on $r_i$.
If $x_i = y_i' = 1$, both algorithms have a cost of $1$ on $r_i$. 
Finally, if $x_i = 1$ and $y_i' = 0$, then $\ALG$ has a cost of $t$ on $r_i$ and $\ALG'$ has a cost of $t$ on $\BLOCK{x_i}{y_i'}$.
\end{proof}

A \emph{star graph} is a tree where all vertices, except at most one, are leaves.
A \emph{star forest} is a disjoint collection of star graphs.
Since the graphs constructed in the proof of Lemma~\ref{lem:bdvc_t_hardness} are star forests, we obtain the following corollary.

\begin{corollary}
\label{corollary_VCSF}
For any $t\in \ZZ^+$ and any pair of insertion monotone error measures, $\PAIRETA$, $\BDVC{t}$ is $\CCWM{\PAIRETAFORCC}{t}$-hard, even on star forests.
\end{corollary}

Since star forests are interval graphs, we also obtain the following.
\begin{corollary}
\label{corollary_VCIG}
For any $t\in \ZZ^+$ and any pair of insertion monotone error measures, $\PAIRETA$, $\BDVC{t}$ is $\CCWM{\PAIRETAFORCC}{t}$-hard, even on interval graphs.
\end{corollary}
Note that Figure~\ref{fig:bounded_degree_vertex_cover_example} is an (interval) graph representation of the intervals in Figure~\ref{fig:interval_scheduling_example}.

Now, we turn to showing membership by reducing to $\ASG{t}$.

\begin{lemma}\label{lem:bdvc_t_containment}
For any $t \in \ZZ^+$ and any pair of error measures, $\PAIRETA$, $\BDVC{t} \in \CCWM{\PAIRETAFORCC}{t}$.
\end{lemma}
\begin{proof}
We define a strict online reduction, $\rho \colon \BDVC{t} \orarrow \ASG{t}$.
Consider any $I = (x,\hat{x},r) \in \INSTANCES{\BDVC{t}}$ and any $\ALG' \in \ALGS{\ASG{t}}$.
We define an instance, $I' =(x',\hat{x}',r') \in \INSTANCES{\ASG{t}}$, and an algorithm, \ALG, for handling $I$ using the output of $\ALG'$ on $I'$:

For each request, $r_i$ in $r$, give $\hat{x}'_i=\hat{x}_i$ to $\ALG'$ and let
$y'_i$ be the output of $\ALG'$.
If $v_i$ has a neighbor among $v_1,\ldots,v_{i-1}$ which is not in the
vertex cover constructed so far, \ALG outputs $y_i=1$,
  ensuring that \ALG always outputs a vertex cover.
Otherwise, it outputs $y_i=y'_i$.
After the last request of $I$, compute an optimal solution, $x$, for $I$ and present $x'=x$ to $\ALG'$, in order to finish the instance $I'$.

Since we let
 $x'=x$ and $\hat{x}'=\hat{x}$, Condition~\ref{item:OR_condition_eta} from Definition~\ref{def:OR} is trivially satisfied for any pair of error measures.
Moreover, since $x'=x$, $\OPT(I) = \OPT(I')$, and so Condition~\ref{item:OR_condition_OPT} is also satisfied.
Hence, it only remains to check Condition~\ref{item:OR_condition_ALG}.
To this end, we prove that $\ALG(I) \leqslant \ALG'(I')$.

We consider the cost of $\ALG$ and $\ALG'$ on request $r_i$:
\begin{itemize}
\item $y'_i=y_i=0$: 
  \begin{itemize}
  \item $x_i=0$: Both algorithms have a cost of $0$.
  \item $x_i=1$: $\ALG'$ has a cost of $t$ and $\ALG$ has a cost of $0$.
  \end{itemize}
\item $y'_i=0$, $y_i=1$: $\ALG$ has a cost of $1$.
  \begin{itemize}
  \item $x_i=0$: $\ALG'$ has a cost of $0$.
  \item $x_i=1$: $\ALG'$ has a cost of $t$.
  \end{itemize}
\item $y'_i=y_i=1$: Both algorithms have a cost of $1$.
\end{itemize}
Note that $\ALG$ has a higher cost than $\ALG'$, only when $x_i=y'_i=0$ and $y_i=1$.
In this case, the cost of $\ALG$ is exactly one higher than that of $\ALG'$.
However, from $y_i \neq y'_i$, it follows by the definition of \ALG that $v_i$ has a neighbor, $v_j$, $j<i$, such that $y_j=0$,
in which case also $y'_j=0$.
Moreover, $x_i=0$ implies that $x_j=1$, since $x$ encodes a feasible vertex cover.
Since $v_j$ has at most $t$ neighbors, and since the cost of $\ALG'$ on $r'_j$ is $t$ higher than that of $\ALG$ on $r_j$, we conclude that the total cost of $\ALG$ is no larger than that of \ALG.
Hence, Condition~\ref{item:OR_condition_ALG} is satisfied.
\end{proof}

Our results about Vertex Cover are summarized in Theorem~\ref{thm:summary_bdvc_t}.
Note that Items~\ref{item:bdvc_t_completeness} and~\ref{item:bdvc_t_membership} follow directly from Lemmas~\ref{lem:bdvc_t_hardness} and~\ref{lem:bdvc_t_containment}.

\begin{theorem}\label{thm:summary_bdvc_t}
  Consider any $t \in \ZZ^+$.

For any pair of insertion monotone error measures, $\PAIRETA$,
\begin{enumerate}[label = {(\roman*)}]
\item $\BDVC{t}$ is $\CCWM{\PAIRETAFORCC}{t}$-complete, \label{item:bdvc_t_completeness}
\item $\BDVC{t}$ on star forests is $\CCWM{\PAIRETAFORCC}{t}$-complete, \label{item:bdvcsf_t_completeness}
\item $\BDVC{t}$ on interval graphs is $\CCWM{\PAIRETAFORCC}{t}$-complete, \label{item:bdvcig_t_completeness}
\item $\VC$ is $\CCWM{\PAIRETAFORCC}{t}$-hard, and $\VC \not\in \CCWM{\PAIRETAFORCC}{t}$. \label{item:vc_insertion_monotone}
\end{enumerate}
For any pair of error measures, $\PAIRETA$,
\begin{enumerate}[resume,label = {(\roman*)}]
\item $\BDVC{t} \in \CCWM{\PAIRETAFORCC}{t}$, \label{item:bdvc_t_membership}
\item $\BDVC{t}$ on star forests is contained in $\CCWM{\PAIRETAFORCC}{t}$, \label{item:bdvcsf_t_membership}
\item $\BDVC{t}$ on interval graphs is contained in $\CCWM{\PAIRETAFORCC}{t}$, \label{item:bdvcig_t_membership}
\item $\VC \in \CCINFTYWM{\PAIRETAFORCC}$, and $\VC$ is not $\CCINFTYWM{\PAIRETAFORCC}$-hard. \label{item:vc}
\end{enumerate}
\end{theorem}
\begin{proof}
First, we consider insertion monotone error measures:
  
\textbf{Towards~\ref{item:bdvc_t_completeness}:}
This is a direct consequence of Lemmas~\ref{lem:bdvc_t_hardness} and~\ref{lem:bdvc_t_containment}.

\textbf{Towards~\ref{item:bdvcsf_t_completeness}:}
This is a direct consequence of Corollary~\ref{corollary_VCSF} and of Lemma~\ref{lem:bdvc_t_containment} in combination with Corollary~\ref{cor:sub_and_sup_problem}\ref{item:sub}.

\textbf{Towards~\ref{item:bdvcig_t_completeness}:}
This is a direct consequence of Corollary~\ref{corollary_VCIG} and of Lemma~\ref{lem:bdvc_t_containment} in combination with Corollary~\ref{cor:sub_and_sup_problem}\ref{item:sub}.

\textbf{Towards~\ref{item:vc_insertion_monotone}:}
The hardness of \VC is a direct consequence of Lemma~\ref{lem:bdvc_t_hardness} in combination with Corollary~\ref{cor:sub_and_sup_problem}\ref{item:sub}.
Now, assume towards contradiction that  $\VC \in \CCWM{\PAIRETAFORCC}{t}$, for some $t \in \ZZ^+$.
This means that \ASG{t} $\geq_{\textrm o}$ \VC.
By Lemma~\ref{lem:bdvc_t_hardness}, \VC $\geq_{\textrm o}$ \ASG{t+1}, and so, by transitivity (Lemma~\ref{lem:transitivity}), $\ASG{t} \geq_{\textrm o} \ASG{t+1}$, which contradicts Lemma~\ref{lem:hierarchy_lemma}\ref{item:no_or_from_asg_t+1_to_asg_t}.

Next, we consider general error measures:

\textbf{Towards~\ref{item:bdvc_t_membership}:}
This is a direct consequence of Lemma~\ref{lem:bdvc_t_containment}.

\textbf{Towards~\ref{item:bdvcsf_t_membership}--\ref{item:bdvcig_t_membership}:}
These are direct consequences of Lemma~\ref{lem:bdvc_t_containment} in combination with Corollary~\ref{cor:sub_and_sup_problem}\ref{item:sub}.

\textbf{Towards~\ref{item:vc}:}
To prove that $\VC \in \CCINFTYWM{\PAIRETAFORCC}$, we give a strict online reduction $\rho \colon \VC \orarrow \ASGINFTY$.
To this end, consider any $\ALG' \in \ALGS{\ASGINFTY}$ and any $I = (x,\hat{x},r) \in \INSTANCES{\VC}$.
Let $I' = \ORTRANS{\rho}(\ALG',I)= (x',\hat{x}',r')$, where
$x'=x$, $\hat{x}'=\hat{x}$, and $r'$ is a sequence of prompts for guessing
the next bit as usual for $\ASGINFTY$. Moreover, let $\ALG = \ORALG{\rho}(\ALG')$ be the algorithm that always outputs the same as $\ALG'$ and, at the end, computes $\OPT_{\BDVC{t}}[I]$ and reveals it to $\ALG'$ in the form of the bit string $x'=x$.

By construction, Conditions~\ref{item:OR_condition_OPT} and~\ref{item:OR_condition_eta} from Definition~\ref{def:OR} are trivially satisfied, and so it only remains to check Condition~\ref{item:OR_condition_ALG}.
Since \ALG outputs the same as $\ALG'$, we only need to argue that as long as $\ALG'$ produces a feasible solution, \ALG does so too.
Hence, suppose that $\ALG$ has created an infeasible solution to instance $I$.
Then, there exists an uncovered edge $(v_i,v_j)$.
Hence, $y_i = y_j = 0$, where $y_i$ and $y_j$ are the $i$th and $j$th guesses made by $\ALG'$.
However, since the edge $(v_i,v_j)$ is contained in the underlying graph of $I$, at least one of $v_i$ and $v_j$ has been accepted by $\OPT_{\VC}$, and so $x_i = 1$ or $x_j = 1$.
Hence, $\ALG'$ has guessed $0$ on a true $1$, implying that $\ALG'$ produced an infeasible solution.

To prove that $\VC$ is not $\CCINFTYWM{\PAIRETAFORCC}$-hard,
let $\textsc{Acc}$ be the following online algorithm for $\VC$. 
When receiving a vertex, $v$, if $v$ has no neighbors, reject $v$; otherwise, accept $v$.
We claim that $\textsc{Acc}$ is $(n-1,0,0)$-competitive, where $n = \abs{V}$.
To this end, consider any graph, $G$.
If $\OPT(G) = 0$, then $G$ contains no edges, and so $\textsc{Acc}$ never accepts a vertex, in which case $\textsc{Acc}(G) = 0$.
If, on the other hand, $\OPT(G) \geqslant 1$, then $\textsc{Acc}(G) \leqslant n-1$, since $\textsc{Acc}$ never accepts the first vertex.
Hence, $\textsc{Acc}$, is $(n-1,0,0)$-competitive.
If $\VC \geq_{\textrm o} \ASGINFTY$, then this implies the existence of an $(n-1,0,0)$-competitive algorithm for $\ASGINFTY$,
contradicting Observation~\ref{obs:no_competitive_algorithm_for_asg_infty}.
\end{proof}

\subsection{Interval Rejection}
\label{sec:IR}

Given a collection of intervals $\mathcal{S}$ on the line, an Interval Rejection algorithm finds a subset $\mathcal{S}' \subseteq \mathcal{S}$ of intervals such that no two intervals in $\mathcal{S} \setminus \mathcal{S}'$ overlap.
The cost of the solution is given by the cardinality of $\mathcal{S}'$, and the goal is to minimize this cost.

In the offline setting, Interval Rejection is solvable in polynomial time, by a greedy algorithm~\cite{KT05}.

\begin{definition}\label{def:ir_t}
A request, $r_i$, for \emph{Online $t$-Bounded Overlap Interval Rejection with Predictions} ($\INTER{t}$) is an interval, $S_i$.
Instances for $\INTER{t}$ satisfy that any requested interval overlaps at most $t$ other intervals in the instance.

An algorithm, $\ALG$, outputs $y_i = 1$ to accept $S_i$ and $y_i=0$ to reject it.
For the algorithm's solution to be feasible, the intervals in $\{S_i \mid y_i=0 \}$ must be pairwise nonoverlapping.
Given an instance $I = (x,\hat{x},r) \in \INSTANCES{\INTER{t}}$, 
$$\ALG(I) = \sum_{i=1}^{n} y_i.$$
The set $\{ S_i \mid x_i=0 \}$ is an optimal solution.
\end{definition}

We also consider Online Interval Rejection with Predictions ($\IR$), where there is no bound on the number of overlaps.

\begin{figure}[htp]
\centering
\begin{tikzpicture}
\node at (-1.6,0.275) {$(x_i,y_i)=(x'_i,y'_i)$:};
\draw[|-|] (0,0.6) -- node[above] {$S_1$} node[below] {$(0,0)$} (2.25,0.6);
\draw[|-|] (2.25,0.6) -- node[above] {$S_2$} node[below] {$(0,1)$} (4.5,0.6);
\draw[|-|] (4.5,0.6) -- node[above] {$S_3$} node[below] {$(1,1)$} (6.75,0.6);
\draw[|-|] (6.75,0.6) -- node[above] {$S_4$} node[below] {$(1,0)$} (9,0.6);
\begin{scope}[shift={(0,-0.5)}]
\draw[|-|] (4.5,-0.5)   -- node[above] {$S_3^1$} (6.75,-0.5);
\draw[|-|] (6.75,-0.5) -- node[above] {$S_4^1$} (7.5,-0.5);
\draw[|-|] (7.5,-0.5) -- node[above] {$S_4^2$} (8.25,-0.5);
\draw[|-|] (8.25,-0.5) -- node[above] {$S_4^3$} (9,-0.5);
\end{scope}
\begin{scope}[shift={(0,-9.7)}]
  \node at (-1.6,7.4) {$\OPT = \OPT'$:}; \node at (1.125,7.4) {$0$}; \node at (3.375,7.4) {$0$}; \node at (5.625,7.4) {$1$};           \node at (7.875,7.4) {$1$}; 
  \node at (-1.6,6.8) {\ALG:};           \node at (1.125,6.8) {$0$}; \node at (3.375,6.8) {$1$}; \node at (5.625,6.8) {$1$};           \node at (7.875,6.8) {$t$}; 
  \node at (-1.6,6.2) {$\ALG'$:};        \node at (1.125,6.2) {$0$}; \node at (3.375,6.2) {$1$}; \node at (5.625,6.2) {$\geqslant 1$}; \node at (7.875,6.2) {$t$};
\end{scope}
\end{tikzpicture}
\caption{An \INTER{3} instance constructed from an \ASG{3} instance as described in the proof of Lemma~\ref{lem:interval_scheduling_hardness}.
  For each challenge request, the pair $(x_i,y'_i)$ is shown below $S_i$.
Below the \INTER{3} instance, the cost of each challenge request and its corresponding block requests is shown for $\OPT$, $\OPT'$, $\ALG$, and $\ALG'$.
}
\label{fig:interval_scheduling_example}
\end{figure}

\begin{lemma}\label{lem:interval_scheduling_hardness}
For any $t \in \ZZ^+$ and any pair of insertion monotone error measures, $\PAIRETA$, $\INTER{t}$ is $\CCWM{\PAIRETAFORCC}{t}$-hard. 
\end{lemma}
\begin{proof}
We prove that $\INTER{t} \geq_{\textrm o} \ASG{t}$ by giving a strict online reduction $\rho \colon \ASG{t} \orarrow \INTER{t}$.
This reduction is essentially identical to the one from $\ASG{t}$ to $\BDVC{t}$ from the proof of Lemma~\ref{lem:bdvc_t_hardness}, since the graph produced there is an
interval graph, but we present it for completeness.
In particular, the $i$'th challenge request is $S_i = ((i-1)\cdot t , i \cdot t)$, the blocks $\BLOCK{0}{0}$ and $\BLOCK{0}{1}$ are empty, the block $\BLOCK{1}{1}$ contains a single request to $S_i^1 = S_i$, and the block $\BLOCK{1}{0}$ contains $t$ requests, $S_i^j$, for $j=1,2,\ldots,t$, where $S_i^j = ((i-1) \cdot t + j - 1, (i-1) \cdot t + j)$.
See Figure~\ref{fig:interval_scheduling_example} for an example reduction.
Observe that this is an interval representation of the interval graph created in the reduction from Lemma~\ref{lem:bdvc_t_hardness}, assuming that the $\BDVC{t}$ algorithm outputs the same bits as the $\INTER{t}$ algorithm.
The remainder of the analysis is analogous to that of Lemma~\ref{lem:bdvc_t_hardness}.
\end{proof}

Now, we turn to showing membership by reducing to $\BDVC{t}$.

\begin{lemma}\label{lem:interval_scheduling_to_bdvc_t}
For any $t \in \ZZ^+ \cup \{\infty\}$ and any pair of error measures, $\PAIRETA$, $\INTER{t} \in \CCWM{\PAIRETAFORCC}{t}$.
\end{lemma}
\begin{proof}
Let $\PAIRETA$ be any pair of error measures.
We define a reduction, $\rho\colon \INTER{t} \orarrow \BDVC{t}$.
This, together with Theorem~\ref{lem:different_base_problem}\ref{item:sub_problems} and Lemma~\ref{lem:bdvc_t_containment}, shows that $\INTER{t} \in \CCWM{\PAIRETAFORCC}{t}$.
Figure~\ref{fig:reduction_from_inter_to_vc} shows an example of this reduction. 
We will see that the reduction resembles a template reduction, with all blocks empty, in that $x'=x$, $\hat{x}'=\hat{x}$, and $y'=y$.

\begin{figure}[htp]
\centering
\begin{tikzpicture}
\begin{scope}[shift={(-1,0)}]
\draw[|-|] (1.25,0.5+1.25-0.6) -- node[above,scale = 1] {$S_1$} (2.66,0.5+1.25-0.6);
\draw[|-|] (2,1.5+1.25) -- node[above,scale = 1] {$S_2$} (3.75,1.5+1.25);
\draw[|-|] (1.75,1+1.25-0.3) -- node[above,scale = 1] {$S_3$} (3,1+1.25-0.3);
\draw[|-|] (3.25,0.5+1.25-0.6) -- node[above,scale = 1] {$S_4$} (4.75,0.5+1.25-0.6);
\draw[|-|] (1.33,0+1.25-0.9) -- node[above,scale = 1] {$S_5$} (5,0+1.25-0.9);
\draw[|-|] (4.5,1+1.25-0.3) -- node[above,scale = 1] {$S_6$} (5.5,1+1.25-0.3);
\end{scope}
\node[vertex] (v1) at (7,3.5) {$v_1$};
\node[vertex] (v2) at (10,3.5) {$v_2$};
\node[vertex] (v3) at (8.5,2.5) {$v_3$};
\node[vertex] (v4) at (10,0) {$v_4$};
\node[vertex] (v5) at (8.5,1) {$v_5$};
\node[vertex] (v6) at (7,0) {$v_6$};

\draw (v1) -- (v2);
\draw (v1) -- (v3);
\draw (v1) -- (v5);

\draw (v2) -- (v3);
\draw (v2) -- (v4);
\draw (v2) -- (v5);

\draw (v3) -- (v5);

\draw (v4) -- (v5);
\draw (v4) -- (v6);

\draw (v5) -- (v6);
\end{tikzpicture}
\caption{An example reduction from $\INTER{t}$ to $\BDVC{t}$, for $t \geqslant 5$.
On the left are the intervals in the instance $I$ of $\INTER{t}$, and on the right is the graph of the instance $\ORTRANS{\rho}(\ALG',I)$, for any algorithm $\ALG' \in \ALGS{\BDVC{t}}$.}
\label{fig:reduction_from_inter_to_vc}
\end{figure}

Consider any algorithm, $\ALG'$, for \BDVC{t} and any instance, $I = (x,\hat{x},r)$, of \INTER{t}.
We define $\ALG = \ORALG{\rho}(\ALG')$ and $I' = (x',\hat{x}',r') = \ORTRANS{\rho}(\ALG',I)$ as follows.
When $\ALG$ receives a request containing an interval, $S_i$, we request a vertex, $v_i$, together with all edges of $\{ (v_j,v_i) \mid j < i \text{ and }S_j \cap S_i \neq \emptyset \}$, with prediction $\hat{x}_i' = \hat{x}_i$.
Since each interval of $I$ overlaps at most $t$ other intervals, $I'$ is a valid instance of \BDVC{t}.
For each $S_i$, $\ALG$ outputs the same as $\ALG'$ does for $v_i$.

Note that the graph, $G$, of $I'$ is the interval graph of the intervals of $I$.
Thus, $\langle y_1,y_2,\ldots,y_{|r|} \rangle$ is a feasible solution to $I$, if and only if $\{v_i \mid y_i=0 \}$ is an independent set in $G$.
Since the complement of an independent set is a vertex cover, $x'=x$ is an optimal solution to $I'$ and $\ALG(I)=\ALG'(I')$, proving Conditions~\ref{item:OR_condition_OPT} and~\ref{item:OR_condition_ALG} of Definition~\ref{def:OR}.
Combining $x'=x$ with $\hat{x}'=\hat{x}$, we also get that Condition~\ref{item:OR_condition_eta} is satisfied.
\end{proof}

Our results about $\INTER{t}$ and $\IR$ are summarized in Theorem~\ref{thm:summary_inter_t}. 

\begin{theorem}\label{thm:summary_inter_t}
Consider any $t \in \ZZ^+$.

For any pair of insertion monotone error measures, $\PAIRETA$,
\begin{enumerate}[label = {(\roman*)}]
\item $\INTER{t}$ is $\CCWM{\PAIRETAFORCC}{t}$-complete, \label{item:inter_t_completeness}
\item $\IR$ is $\CCWM{\PAIRETAFORCC}{t}$-hard, and $\IR \not\in \CCWM{\PAIRETAFORCC}{t}$. \label{item:IR}
\end{enumerate}
For any pair of error measures, $\PAIRETA$,
\begin{enumerate}[label = {(\roman*)},resume]
\item $\INTER{t} \in \CCWM{\PAIRETAFORCC}{t}$, \label{item:inter_t_member_all_errors}
\item $\IR \in \CCINFTYWM{\PAIRETAFORCC}$, and $\IR$ is not $\CCINFTYWM{\PAIRETAFORCC}$-hard. \label{item:interval_scheduling}
\end{enumerate}
\end{theorem}
\begin{proof}
  We prove each item separately.
  
\textbf{Towards~\ref{item:inter_t_completeness}:}
This is a direct consequence of Lemmas~\ref{lem:interval_scheduling_hardness} and~\ref{lem:interval_scheduling_to_bdvc_t}.

\textbf{Towards~\ref{item:IR}:}
Since $\INTER{t}$ is a subproblem of $\IR$, the $\CCWM{\PAIRETAFORCC}{t}$-hardness of \IR is a direct consequence of Corollary~\ref{cor:sub_and_sup_problem}\ref{item:sup} and Lemma~\ref{lem:interval_scheduling_hardness}.

Assume towards contradiction that $\IR \in \CCWM{\PAIRETAFORCC}{t}$.
Then $\ASG{t} \geq_{\textrm o} \IR$.
On the other hand, $\IR \geq_{\textrm o} \ASG{t+1}$, so transitivity implies that $\ASG{t} \geq_{\textrm o} \ASG{t+1}$, which contradicts Lemma~\ref{lem:hierarchy_lemma}\ref{item:no_or_from_asg_t+1_to_asg_t}.

\textbf{Towards~\ref{item:inter_t_member_all_errors}:}
This is a direct consequences of Lemma~\ref{lem:interval_scheduling_to_bdvc_t}.

\textbf{Towards~\ref{item:interval_scheduling}:}
By Lemma~\ref{lem:interval_scheduling_to_bdvc_t}, $\IR \leq_{\textrm o} \VC$, and  by Theorem~\ref{thm:summary_bdvc_t}\ref{item:vc}, $\VC \in \CCINFTYWM{\PAIRETAFORCC}$. Hence, by transitivity, $\IR \in \CCINFTYWM{\PAIRETAFORCC}$.

Assume towards contradiction that $\IR$ is $\CCINFTYWM{\PAIRETAFORCC}$-hard.
Then, $\IR \geq_{\textrm o} \ASGINFTY$.
Since $\VC \geq_{\textrm o} \IR$, transitivity implies that $\VC \geq_{\textrm o} \ASGINFTY$,
contradicting Theorem~\ref{thm:summary_bdvc_t}\ref{item:vc}. 
\end{proof}

\subsection{$\boldsymbol{k}$-Spill}\label{sec:SPILL}

Given a graph, $G = (V,E)$, the objective of $k$-Spill is to
select a smallest possible subset $\xOne{V} \subseteq V$ such that the subgraph of $G$ induced by the vertices in $V \setminus \xOne{V}$ is $k$-colorable.

In compiler construction~\cite{A98}, register allocation plays a significant role.
It is often implemented by a liveness analysis, followed by a construction of an interference graph, which is then colored. 
The vertices represent variables (or values) and the edges represent conflicts (values that must be kept at the same point in time). 
The goal is to place as many of these values as possible in registers. 
Thus, with a fixed number of registers, this is really coloring with a fixed number of colors. 
Vertices that cannot be colored are referred to as \emph{spills}. 
Spilled values must be stored in a more expensive manner. 
Thus, similar to minimizing faults in Paging, the objective is to minimize spill.

\begin{figure}[htp]
\centering
\begin{tikzpicture}[scale = 0.89]

\begin{scope}[shift={(-3,0)}]
\node at (5, 0.65) {$0$};
\node[vertex] (v1) at (5,0) {$v_1$};
\node at (5,-0.7) {$0$};
\node[vertex] (f1) at (6.5,0) {$f_1$};
\draw (v1) -- (f1);
\end{scope}

\begin{scope}[shift={(-4.5,0)}]
\node at (9.65, 0.65) {$1$};
\node[vertex] (v2) at (9.5,0) {$v_2$};
\node at (9.5,-0.7) {$1$};
\draw (f1) -- (v2);
\node[vertex] (v21) at (8.4,1.25) {$v_{2,1}$};
\node[vertex] (v22) at (9.5,2.5) {$v_{2,2}$};
\node[vertex] (v23) at (10.6,1.25) {$v_{2,3}$};
\node[vertex] (f2) at (11,0) {$f_2$};
\draw (v2) -- (v21);
\draw (v2) -- (v22);
\draw (v2) -- (v23);
\draw (v21) -- (v23);
\draw (v22) -- (v23);
\draw (v21) -- (v22);
\draw (v2) -- (f2);
\end{scope}

\begin{scope}[shift={(3,0)}]
\node at (5, 0.65) {$0$};
\node[vertex] (v3) at (5,0) {$v_3$};
\node at (5,-0.7) {$1$};
\draw (f2) -- (v3);
\node[vertex] (f3) at (6.5,0) {$f_3$};
\draw (v3) -- (f3);
\end{scope}

\begin{scope}[shift={(9,0)}]
\node at (2.15, 0.65) {$1$};
\node[vertex] (v4) at (2,0) {$v_4$};
\node at (2,-0.7) {$0$};
\draw (f3) -- (v4);
\node[vertex] (v41) at (0.3,1.55) {$v_{4,1}$};
\node at (0.3,0.85) {$0$};) 
\node[vertex] (v42) at (0.3,3.5) {$v_{4,2}$};
\node at (0.15,2.8) {$1$};) 
\node[vertex] (v43) at (2,4.75) {$v_{4,3}$};
\node at (2.15,4.05) {$0$};) 
\node[vertex] (v44) at (3.7,3.5) {$v_{4,4}$};
\node at (3.7,2.8) {$1$};) 
\node[vertex] (v45) at (3.7,1.55) {$v_{4,5}$};
\node at (3.7,0.85) {$1$};) 
\node[vertex] (f4) at (3.5,0) {$f_4$};
\draw (v4) -- (v41);
\draw (v4) --(v42);
\draw (v41) --(v42);
\draw (v4) -- (v43);
\draw (v41) -- (v43);
\draw (v4) -- (v44);
\draw (v41) -- (v44);
\draw (v43) -- (v44);
\draw (v4) -- (v45);
\draw (v41) -- (v45);
\draw (v43) -- (v45);
\draw (v4) -- (f4);
\end{scope}

\begin{scope}[shift={(1.85,-9.3)}]
  \node at (-1.8,7.4) {$\OPT = \OPT'$:}; \node at (0.15,7.4) {$0$}; \node at (3.15,7.4) {$1$};           \node at (6.15,7.4) {$0$}; \node at (9.15,7.4) {$1$}; 
  \node at (-1.8,6.8) {\ALG:};           \node at (0.15,6.8) {$0$}; \node at (3.15,6.8) {$1$};           \node at (6.15,6.8) {$1$}; \node at (9.15,6.8) {$t$}; 
  \node at (-1.8,6.2) {$\ALG'$:};        \node at (0.15,6.2) {$0$}; \node at (3.15,6.2) {$\geqslant 1$}; \node at (6.15,6.2) {$1$}; \node at (9.15,6.2) {$t$};
\end{scope}
\end{tikzpicture}
\caption{A \COLT{3}{7} instance constructed from an \ASG{3} instance as described in the proof of Theorem~\ref{thm:k-col_hardness}.
  For each challenge request, the bits $x_i$ and $y'_i$ are shown above and below $v_i$, respectively.
  For the block requests, only the $y_i'$ bits influencing the construction of the graph are shown.
}
\label{fig:k-col_example}
\end{figure}

\begin{figure}[htp]
\centering
\begin{tikzpicture}[scale = 0.89]
  
\begin{scope}[shift={(-3,0)}]
\node[vertex] (v1) at (5,0) {$v_1$};
\node[vertex] (f1) at (6.5,0) {$f_1$};
\draw (v1) -- (f1);
\end{scope}

\begin{scope}[shift={(-4.5,0)}]
\node[vertexgray] (v2) at (9.5,0) {$v_2$};
\draw[gray!30] (f1) -- (v2);
\node[vertex] (v21) at (8.4,1.25) {$v_{2,1}$};
\node[vertex] (v22) at (9.5,2.5) {$v_{2,2}$};
\node[vertex] (v23) at (10.6,1.25) {$v_{2,3}$};
\node[vertex] (f2) at (11,0) {$f_2$};
\draw[gray!30] (v2) -- (v21);
\draw[gray!30] (v2) -- (v22);
\draw[gray!30] (v2) -- (v23);
\draw (v21) -- (v23);
\draw (v22) -- (v23);
\draw (v21) -- (v22);
\draw[gray!30] (v2) -- (f2);
\end{scope}

\begin{scope}[shift={(3,0)}]
\node[vertex] (v3) at (5,0) {$v_3$};
\draw (f2) -- (v3);
\node[vertex] (f3) at (6.5,0) {$f_3$};
\draw (v3) -- (f3);
\end{scope}

\begin{scope}[shift={(9,0)}]
\node[vertexgray] (v4) at (2,0) {$v_4$};
\draw[gray!30] (f3) -- (v4);
\node[vertex] (v41) at (0.3,1.55) {$v_{4,1}$};
\node[vertex] (v42) at (0.3,3.5) {$v_{4,2}$};
\node[vertex] (v43) at (2,4.75) {$v_{4,3}$};
\node[vertex] (v44) at (3.7,3.5) {$v_{4,4}$};
\node[vertex] (v45) at (3.7,1.55) {$v_{4,5}$};
\node[vertex] (f4) at (3.5,0) {$f_4$};
\draw[gray!30] (v4) -- (v41);
\draw[gray!30] (v4) --(v42);
\draw (v41) --(v42);
\draw[gray!30] (v4) -- (v43);
\draw (v41) -- (v43);
\draw[gray!30] (v4) -- (v44);
\draw (v41) -- (v44);
\draw (v43) -- (v44);
\draw[gray!30] (v4) -- (v45);
\draw (v41) -- (v45);
\draw (v43) -- (v45);
\draw[gray!30] (v4) -- (f4);
\end{scope}
\end{tikzpicture}
\caption{The graph $G$ from Figure~\ref{fig:k-col_example}. The vertices and edges drawn in black are those of the subgraph $G_0$.
}
\label{fig:k-col_example_G0}
\end{figure}

\begin{definition}\label{def:k-spill}
For \emph{Online $d$-Bounded $k$-Spill with Predictions} ($\COLT{k}{d}$), the input is a graph, $G=(V,E)$, where all vertices have degree at most~$d$.
For each vertex, $v_i$, an algorithm, $\ALG$, outputs $y_i = 1$ to mark $v_i$ as a spill or $y_i=0$ to mark it as colorable.

For the algorithm's solution to be feasible, the subgraph induced by $\{v_i \mid y_i=0 \}$ must be $k$-colorable.
Given an instance $I = (x,\hat{x},r) \in \INSTANCES{\COL{k}}$,
$$\ALG(I) = 
\sum_{i=1}^n y_i.$$
The set $\{ v_i \mid x_i=0 \}$ is a maximum set of vertices inducing a $k$-colorable subgraph.
\end{definition}

We let Online $k$-Spill with Predictions ($\COL{k}$) denote the version where there is no bound on the degree of vertices.

For a fixed number of colors, minimizing spill is equivalent to maximizing the number of colored vertices and this problem is
NP-complete~\cite{AHU74}.
A multi-objective variant of this problem was studied both online and
offline in~\cite{ELW11}.
As observed in the proof of the following theorem,
$\COL{1}$ is equivalent to $\VC$.

\begin{theorem}\label{thm:k-col_hardness}
For all $k,t \in \ZZ^+$, and all pairs of insertion monotone error measures $\PAIRETA$, $\COLT{k}{t+k+1}$ is $\CCWM{\PAIRETAFORCC}{t}$-hard.
\end{theorem}
\begin{proof}
  For any solution, $y'$, to an instance of $\COL{1}$, the vertices of $\{v_i \in V \mid y_i' = 0\}$ form an independent set.
  Thus, $\{v_i \in V \mid y_i' = 1\}$ is a vertex cover, so $\COL{1}$ is equivalent to $\VC$.
  Hence, $\COL{1}$ is $\CCWM{\PAIRETAFORCC}{t}$-hard by Theorem~\ref{thm:summary_bdvc_t}\ref{item:vc_insertion_monotone}.

For $k \geqslant 2$, we give a reduction, $\rho = \ONLRED{\rho}$, from \ASG{t} to \COLT{k}{t+k+1} with respect to $\PAIRETA$, using the reduction template of Algorithm~\ref{alg:redalg} from Section~\ref{sec:reduction_template} (see Figure~\ref{fig:k-col_example} for an example).
To this end, consider any $I = (x,\hat{x},r) \in \INSTANCES{\ASG{t}}$ and any $\ALG' \in \ALGS{\COLT{k}{t+k+1}}$.
Let $\ALG = \ORALG{\rho}(\ALG')$ and $I' = (x',\hat{x}',r')=\ORTRANS{\rho}(\ALG',I)$  be the \ASG{t} algorithm and \COLT{k}{t+k+1} instance created in the reduction described below.
We denote the graph of $I'$ by $G=(V,E)$ and split $V$ in two sets, $\xZero{V} = \{ v_i \in V \mid y_i=0 \}$ and $\xOne{V}  = \{ v_i \in V \mid y_i=1 \}$.

The $i$th challenge request is an isolated vertex, $v_i$, with 
$\hat{x}_i' = \hat{x}_i$.
Let $y_i'$ be the output of $\ALG'$ for $v_i$.
Then, \ALG outputs $y_i=y'_i$ on the $i$th request.

The $i$th block, $\BLOCK{x_i}{y_i'}$, is constructed as follows, with all prediction bits equal to $0$.
The last vertex, $f_i$, of the block is called a \emph{final} vertex.
If $i < n$, $f_i$ is connected to $v_i$ and $v_{i+1}$, and if $i = n$, $f_i$ is connected to $v_i = v_n$.
\begin{itemize}
\item If $x_i = 0$, then $\BLOCK{x_i}{y_i'}$ contains a single vertex:
  \begin{itemize}
  \item The final vertex, $f_i$.
  \end{itemize}
\item If $x_i = y_i' = 1$, then $\BLOCK{x_i}{y_i'}$ contains $k+1$
  vertices:
  \begin{itemize}
  \item For each $j=1,2,\ldots,k$, a vertex, $v_{i,j}$, connected to $v_i$ and to $v_{i,l}$, $1 \leqslant l < j$.
  \item The final vertex, $f_i$.
  \end{itemize}
\item If $x_i=1$ and $y_i'=0$, then $\BLOCK{x_i}{y_i'}$ contains between $k$ and $t+k$ vertices:

  Let $j = 1$.\\
  Until $\ALG'$ has added $t$ new vertices to $V_1$ or $k$ new vertices to $V_0$:\footnote{The latter happens only if $\ALG'$ creates an infeasible solution, incurring an infinite cost.}
\begin{itemize}[ ]
\item A vertex, $v_{i,j}$, connected to $v_i$ and to each vertex in $\xZero{V} \cap \bigcup_{l=1}^{j-1}v_{i,l}$.\\
  Increment $j$.
\end{itemize}
Until $\bigcup_{j} \{v_{i,j}\}$ contains a $k$-clique:
  \begin{itemize}
  \item  A vertex, $v_{i,j}$, connected to $v_i$ and to $v_{i,l}$, $l < j$.\\
    Increment $j$.
  \end{itemize}
  Finally:
  \begin{itemize}
  \item The final vertex, $f_i$.
  \end{itemize}
\end{itemize}

First, we argue that the maximum degree in the graph $G$ is at most $t+k+1$,
Since each challenge request is a neighbor of all corresponding block requests, the maximum degree is obtained by a challenge request.
A challenge request, $v_i$, has degree at most $2$, if $x_i=0$, and at most $k+2$, if $x_i=y_i'=1$.
Now, it only remains to consider the case $x_i=1$ and $y'_i=0$.
In the first step,
where vertices are connected to $v_i$ and to vertices from $V_0$ in the same block, at most $t+k-1$ vertices are added.
In the next step, vertices may be added to ensure a $k$-clique within the block.
If vertices are added in this step, they form a clique with the vertices added to $V_0$ in the previous step.
If the last vertex of the previous step was a $V_1$-vertex, this vertex will also be part of the clique, so the number of $V_0$-vertices from that step plus the number of vertices added to ensure a $k$-clique will be at most $k-1$.
If, on the other hand, the last vertex of the previous step was a $V_0$-vertex, at most $t-1$ $V_1$-vertices were added in that step.
We conclude that the total number of nonfinal vertices in the block is at most $t+k-1$.
Hence, $v_i$ is connected to at most $t+k-1$ nonfinal block vertices, one final vertex of its own block, and at most one final vertex of another block, giving a total number of neighbors of at most $t+k+1$.

Let $X_1 = \{ v_i \in V \mid x_i=1 \}$, i.e., let $X_1$ consist of the challenge requests for which the corresponding block contains more than just a final vertex.
In the example of Figure~\ref{fig:k-col_example}, $X_1=\{v_2,v_4\}$.
Let $X_0 = V \setminus X_1$, and note that $X_0$ contains all block requests.
Let $G_0$ be the subgraph of $G$ induced by $X_0$.
Figure~\ref{fig:k-col_example_G0} shows the subgraph $G_0$ of the example graph of Figure~\ref{fig:k-col_example}.
We argue that $G_0$ is a largest possible $k$-colorable subgraph of $G$.

To see that $G_0$ is $k$-colorable, note that each connected component in $G_0$ is either a path (induced by challenge requests with $x_i=0$ and final vertices) or is induced by (a subset of) the nonfinal vertices of a block corresponding to a challenge request with $x_i=1$.
Since the paths can be $2$-colored, we focus on the nonfinal block vertices.
If $y'_i=1$, the nonfinal vertices of the block induce a $k$-clique.
If $y'_i=0$, each nonfinal vertex, at its arrival, is connected to each vertex, $v_{i,j} \in V_0$.
Since the block ends if it contains at least $k$ vertices in $V_0$, each vertex has at most $k-1$ neighbors in the component at its arrival.
Thus, if the vertices of $G_0$ are colored in order of arrival, one of the $k$ colors is available for each new vertex.

To see that $G$ has no $k$-colorable subgraph with more vertices than $G_0$, note that each vertex in $X_1$ together with its block requests induce a subgraph containing a $(k+1)$-clique.

Since $G_0$ is a largest possible $k$-colorable subgraph of $G$, we conclude that $\OPT(I') = |X_1| = \OPT(I)$, satisfying Condition~\ref{item:OR_condition_OPT} of Definition~\ref{def:OR}.
We also conclude that all block requests are correctly predicted.
Since $\hat{x}'_i=\hat{x}_i$, for each challenge request, $v_i$, and $\eta_0$ and $\eta_1$ are insertion monotone, this implies that $\eta_0(I') \leqslant \eta_0(I)$ and $\eta_1(I') \leqslant \eta_1(I)$, satisfying Condition~\ref{item:OR_condition_eta}.

It only remains to show that Condition~\ref{item:OR_condition_ALG} is satisfied.
To this end, we prove that, for each request, $r_i$ in $I$, the cost of \ALG on $r_i$ is no larger than the cost of $\ALG'$ on $v_i$ and $\BLOCK{x_i}{y_i'}$:

If $y'_i=1$, $\ALG$ guesses $1$, incurring a cost of $1$. $\ALG'$ lets $v_i$ be a spill vertex and thus incurs a cost of at least $1$.

If $y'_i = x_i = 0$, $\ALG$ correctly guesses $0$, incurring a cost of $0$.

If $y'_i = 0$ and $ x_i = 1$, $\ALG$ incorrectly guesses $0$ and incurs a cost of $t$. 
Unless $\ALG'$ produces an infeasible solution, $\BLOCK{x_i}{y_i'}$ contains $t$ vertices of $V_1$, so $\ALG'$ incurs a cost of at least $t$.
\end{proof}

\subsection{$\boldsymbol{2}$-SAT Deletion}

Given a collection of variables, $V$, and a $2$-CNF-SAT formula $\varphi$, an algorithm for $2$-SAT Deletion finds a subset of clauses in $\varphi$ to delete, in order to make $\varphi$ satisfiable~\cite{MR99,CC07}.
Equivalently, the algorithm must assign a truth value to all variables, while minimizing the number of unsatisfied clauses in $\varphi$.

Chleb{\'{\i}}k and Chleb{\'{\i}}kova proved the problem APX-hard~\cite{CC07}.

We consider an online variant of $2$-SAT Deletion in the context of predictions.

\begin{definition}\label{def:2-SAT}
A request, $r_i$, for \emph{Online $2$-SAT Deletion with Predictions} ($\SAT{2}$) contains a variable, $v_i$, and all clauses in the formula of the form $(v_i \vee v_j)$, $(\overline{v}_i \vee v_j)$, $(v_i \vee \overline{v}_j)$, or $(\overline{v}_i \vee \overline{v}_j)$, for all $j < i$.
An algorithm, $\ALG$, outputs $y_i = 1$ to set $v_i = \texttt{true}$.

Given an instance $I = (x,\hat{x},r)$,
\begin{align*}
\ALG(I) = | \{C \in \varphi \mid \mbox{$C$ is unsatisfied}\} |.
\end{align*}
The truth assignment given by $v_i = \texttt{true} \Leftrightarrow x_i=1$ maximizes the number of satisfied clauses.
\end{definition}

\begin{theorem}\label{thm:2_sat_hardness}
For any pair of insertion monotone error measures, $\PAIRETA$, and any $t \in \ZZ^+$, $\SAT{2}$ is $\CCWM{\PAIRETAFORCC}{t}$-hard.
\end{theorem}
\begin{proof}
We define a reduction, $\rho \colon \IR \orarrow \SAT{2}$.
Together with Lemmas~\ref{lem:transitivity} and~\ref{lem:interval_scheduling_hardness}, this establishes the hardness of $\SAT{2}$ with respect to any pair of insertion monotone error measures, \PAIRETA.
The reduction is given in Algorithm~\ref{alg:inter_to_2-sat} and Figure~\ref{fig:inter_to_2-sat} shows an example of the reduction.

\begin{algorithm}[h]
\caption{Reducing \IR to \SAT{2}} 
\begin{algorithmic}[1]
\Require An $\IR$ instance, $I = (x,\hat{x},r)$, and a \SAT{2} algorithm, $\ALG'$
\Ensure A $\SAT{2}$ instance, $I'=(x',\hat{x}',r')$, and an \IR algorithm, \ALG
\State $\varphi \leftarrow \texttt{true}$ \Comment{An empty CNF-formula for $\ALG'$}
\For {each request in $r$, with interval $S_i$ and prediction $\hat{x}_i$}
	\State $O \leftarrow \{ S_j \mid j<i \text{ and } S_j \cap S_i \neq \emptyset \}$ \Comment{Intervals overlapping $S_i$} \label{line:varphi_construction_start}
	\State Create a new variable, $v_{i}$, for $I'$
	\State $\displaystyle \varphi \leftarrow \varphi \, \wedge \, \left(\overline{v}_{i} \vee \overline{v}_{i} \right) \, \wedge \, \left( \bigwedge_{S_j \in O} (v_{i} \vee v_{j})\right)$ \label{line:varphi_construction_end}
	\State Request $v_{i}$ with prediction $\hat{x}_i'=\hat{x}_i$ and updated CNF-formula $\varphi$
	\If {$O$ contains an accepted interval}\label{line:overlap}
		\State Output $y_i=1$ \Comment{Reject $S_i$} \label{line:reject_safety_inter_to_2-sat}
	\Else
		\State Output $y_i=y_i'$, where $y'_i$ is the output of $\ALG'$ for $v_i$\label{line:outputy_i}
	\EndIf
\EndFor
\end{algorithmic}
\label{alg:inter_to_2-sat}
\end{algorithm}

\begin{figure}[htp]
\centering
\begin{tikzpicture}
\draw[|-|] (0,0) -- node[above,scale = 0.8] {$S_1$} (4,0);
\draw[|-|] (2.5,1.5) -- node[above,scale = 0.8] {$S_2$} (3.5,1.5);
\draw[|-|] (3,0.75) -- node[above,scale = 0.8] {$S_3$} (6,0.75);
\draw[|-|] (4.75 , 1.5 ) -- node[above,scale = 0.8 ] {$S_4$} (6.75 , 1.5 );
\end{tikzpicture}
\caption{An example reduction from $\IR$ to $\SAT{2}$.
The above instance of $\IR$ gives rise to the CNF-formula $\varphi = (\overline{v_{1}} \vee \overline{v_{1}}) \wedge (\overline{v_{2}} \vee \overline{v_{2}}) \wedge (v_{1} \vee v_{2}) \wedge (\overline{v_{3}} \vee \overline{v_{3}}) \wedge (v_{1} \vee v_{3}) \wedge (v_{2} \vee v_{3}) \wedge (\overline{v_{4}} \vee \overline{v_{4}}) \wedge (v_{3} \vee v_{4})$.
}
\label{fig:inter_to_2-sat}
\end{figure}

Clauses of the form $(\overline{v}_{i} \vee \overline{v}_{i})$ are called \emph{interval clauses} and clauses of the form $(v_{i} \vee v_{j})$ are called \emph{collision clauses}.

We first prove that, for each request, $r_i$, the cost of \ALG on $S_i$ is no larger than the cost of $\ALG'$ on $v_i$, thus satisfying Condition~\ref{item:OR_condition_ALG} of Definition~\ref{def:OR}.
To this end, note that \ALG incurs a cost of $1$ on $S_i$, if it outputs $y_i=1$ to reject the interval, and otherwise, it incurs no cost.
Moreover, \ALG outputs $y_i=1$, only if at least one of the following two conditions is fulfilled.
\begin{itemize}
\item $\ALG'$ outputs $y'_i=1$ to set $v_i=\texttt{true}$ (see Line~\ref{line:outputy_i} of Algorithm~\ref{alg:inter_to_2-sat}).
  In this case,
  the interval clause $(\overline{v}_i \vee \overline{v}_i)$ is not satisfied, and $\ALG'$ incurs a cost of~$1$.
\item $S_i$ overlaps an interval, $S_j$, $j<i$, which \ALG accepted by outputting $y_j=0$ (see Lines~\ref{line:overlap}--\ref{line:reject_safety_inter_to_2-sat} of Algorithm~\ref{alg:inter_to_2-sat}).
Note that Line~\ref{line:outputy_i} is the only place in Algorithm~\ref{alg:inter_to_2-sat} where \ALG may output $y_j=0$.
Thus, we must have $y'_j=0$, and hence, $v_j=\texttt{false}$.
Therefore, if $v_i=\texttt{false}$, the collision clause $(v_j \vee v_i)$ is not satisfied, giving $\ALG'$ a cost of~$1$.
On the other hand, if $v_i=\texttt{true}$, the interval clause $(\overline{v}_i \vee \overline{v}_i)$ is not satisfied, and $\ALG'$ incurs a cost of~$1$.
\end{itemize}

Next, we note that for any feasible solution, $z$, to $I$, $z'=z$ is a solution to $I'$ satisfying all collision clauses and all interval clauses corresponding to accepted intervals.
Thus, $z'=z$ is a solution to \SAT{2} with the same cost as $z$.
This shows that $\OPT(I') \leqslant \OPT(I)$, satisfying Condition~\ref{item:OR_condition_OPT}.

To prove that Condition~\ref{item:OR_condition_eta} is satisfied, we show that $x'=x$ is an optimal solution to $I'$.
Since any solution, $z'$, to $I'$ satisfying all collision clauses has the same cost as the solution $z=z'$ for $I$, $x'=x$ is optimal among the solutions satisfying all collision clauses.
Thus, we just need to show that no solution to $I'$ violating some collision clause has a lower cost than $x'=x$.
To this end, we argue that, for any solution, $z'$, to $I'$, there is a solution, $u'$, satisfying all collision clauses and with a cost no higher than that of $z'$.
The solution $u'$ can be constructed in the following way. As long as there is an unsatisfied collision clause, change the truth value of a variable of that clause to \texttt{true}. In this way, the collision clause becomes satisfied, no collision clause changes from \texttt{true} to \texttt{false}, and only one interval clause changes from \texttt{true} to \texttt{false}, so the cost does not increase.
Since the number of satisfied collision clauses increases in each step, the process will terminate, and the resulting solution has a cost no higher than that of $z'$.
Since $\hat{x}=\hat{x}'$, this completes the proof that Condition~\ref{item:OR_condition_eta} is satisfied
\end{proof}

\begin{corollary}
For all $t \in \ZZ^+$, and all pairs of insertion monotone error measures $\PAIRETA$, $\SAT{2} \not\in \CCWM{\PAIRETAFORCC}{t}$.
\end{corollary}
\begin{proof}
Assume that $\SAT{2} \in \CCWM{\PAIRETAFORCC}{t}$ for some $t$ and some pair of error measures $\PAIRETA$.
Then, $\ASG{t} \geq_{\textrm o} \SAT{2}$ and, by Theorem~\ref{thm:2_sat_hardness}, $\SAT{2} \geq_{\textrm o} \ASG{t+1}$.
By transitivity, $\ASG{t} \geq_{\textrm o} \ASG{t+1}$, contradicting Lemma~\ref{lem:hierarchy_lemma}\ref{item:no_or_from_asg_t+1_to_asg_t}.
\end{proof}

\subsection{Dominating Set}

Given a graph, $G = (V,E)$, an algorithm for Dominating Set finds a  subset $V' \subseteq V$ of vertices such that for all vertices $v \in V$, $v \in V'$ or $v$ is adjacent to a vertex $u \in V'$.
The cost of the solution is given by the size of $V'$, and the goal is to minimize this cost.

Dominating Set was one of Karp's first 21 NP-complete problems~\cite{K72}, though he referred to it as the set cover problem.
Previous work has concluded that Dominating Set is W[2]-complete~\cite{DF95} and APX-hard~\cite{DP06}.
Further, Online Dominating Set is studied in~\cite{BEFKL19j} and shown to be AOC-complete in~\cite{BFKM17}.

We study an online variant of Dominating Set with predictions.

\begin{definition}\label{def:dom}
The input to \emph{Online Dominating Set with Predictions} ($\DOM$) is a graph, $G=(V,E)$.
An algorithm, $\ALG$, outputs $y_i = 1$ to accept $v_i$ into its dominating set.

For the algorithm's solution to be feasible, the accepted vertices must form a dominating set, i.e., each vertex in $V$ must be contained in $\{v_i \mid y_i=1\}$ or have a neighbor in $\{v_i \mid y_i=1\}$.
Given an instance $I \in \INSTANCES{\DOM}$,
$$\ALG(I) = \sum_{i=1}^n y_i.$$
The set $\{ v_i \mid x_i=1 \}$ is an optimal dominating set.
\end{definition}

\subsubsection{Asymptotic Hardness}

To describe our most central result on $\DOM$, we introduce a slightly weaker
notion of hardness, considering only algorithms that are $(O(1),\beta,\gamma)$-com\-pe\-ti\-tive:

\begin{definition}\label{def:weaklyhard}
Let $P$ and $Q$ be online problems with predictions.
If any $\ABC$-competitive algorithm for $Q$, with $\ALPHA \in O(1)$, implies the existence of an $\ABC$-competitive algorithm for $P$, then $Q$ is \emph{asymptotically as hard as} $P$.
If $P$ is \CCWM{}{}-hard for some complexity class, \CCWM{}{}, and $Q$ is asymptotically as hard as $P$, then $Q$ is \emph{asymptotically \CCWM{}{}-hard}.
\end{definition}

For proving asymptotic hardness, we define a new,
less restrictive type of online reduction, which
motivated the definition of asymptotically as
hard as:
\begin{definition}\label{def:asympreduc}
An \emph{asymptotic online reduction} is the same as a strict online reduction except that, instead of Condition~\ref{item:OR_condition_OPT} from Definition~\ref{def:OR}, it satisfies the following weaker condition:
\begin{align}
\OPT_Q(I_Q) \leqslant \OPT_P(I_P) + b, \text{ for some } b \in O(1)  \label{eq:alternative_O2} \tag{O2$'$}
\end{align}
\end{definition}

We show that, for $(O(1),\beta,\gamma)$-competitive algorithms, asymptotic online reductions can be used the same way as strict online reductions:

\begin{lemma}\label{lem:why_the_naming_asymptotic}
  Let $\rho = \ONLRED{\rho}$ be an asymptotic online reduction from a problem, $P$, with predictions and error measures \PAIRETA to a problem, $Q$, with predictions and error measures \PAIRPHI.
  Let $\ALG_Q \in \ALGS{Q}$ be an $\ABC$-competitive algorithm for $Q$ with respect to $\PAIRPHI$, with $\alpha \in O(1)$, and let $\ALG_P = \ORALG{\rho}(\ALG_Q)$.
  Then, $\ALG_P$ is an $\ABC$-competitive algorithm for $P$ with respect to $\PAIRETA$.
\end{lemma}
\begin{proof}
Consider any instance, $I_P \in \INSTANCES{P}$, and let $I_Q = \ORTRANS{\rho}(\ALG_Q,I_P)$.
Then,
\begin{align*}
  \ALG_P(I_P)
  \leqslant & \ALG_Q(I_Q) + a, \\
            & \text{ where } a \in O(1), \text{ by Condition~\ref{item:OR_condition_ALG} of Definition~\ref{def:OR}}   \\ 
  \leqslant & \: \ALPHA \cdot \OPT_Q(I_Q) + \BETA \cdot \PHIZERO(I_Q) + \GAMMA \cdot \PHIONE(I_Q) + \AT + a, \\
            & \text{ where } \AT \in O(1),\text{ by Definition~\ref{def:competitiveness}}\\
  \leqslant & \: \ALPHA \cdot (\OPT_P(I_P)+b) + \BETA \cdot \ETAZERO(I_P) + \GAMMA \cdot \ETAONE(I_P) + \AT + a, \\
            & \text{ where } b \in O(1), \text{ by~\eqref{eq:alternative_O2} and~\ref{item:OR_condition_eta}} \\
          = & \: \ALPHA \cdot \OPT_P(I_P) + \BETA \cdot \ETAZERO(I_P)+ \GAMMA \cdot \ETAONE(I_P) + \alpha \cdot b + \AT + a,\\
            & \text{ where } \alpha \cdot b + \AT + a \in O(1), \text{ since } \alpha \in O(1).
\end{align*}
\end{proof}

\subsubsection{Hardness of \DOM}
In this section, we give a hardness result for \DOM, using a reduction analogous to a well-known reduction from Vertex Cover to Dominating Set~\cite{M89} used for proving Dominating Set NP-complete.

For any graph, $G=(V,E)$, we let $\dom{G} = (V',E')$ be the supergraph of $G$, where 
\begin{itemize}
\item $V' = V \cup \{ v_{i,j} \mid (v_i,v_j) \in E \}$, and
\item $E' = E \cup \{ (v_i,v_{i,j}), (v_{i,j},v_j) \mid (v_i,v_j) \in E \}$.
\end{itemize}
See Figure~\ref{fig:dom_example} for an example of \dom{G}.

\begin{figure}[htp]
\centering
\begin{tikzpicture}
\node at (-1,4.1) {$G\colon$};
\node[vertex,thick] (v1) at (0,2) {$v_1$};
\node[vertex] (v2) at (0,4) {$v_2$};
\node[vertex] (v3) at (0,0) {$v_3$};
\draw (v1) -- (v2);
\draw (v1) -- (v3);
\end{tikzpicture}
\hspace{1.5cm}
\begin{tikzpicture}
\node[vertex,thick] (v1) at (0,2) {$v_1$};
\node[vertex] (v2) at (0,4) {$v_2$};
\node[vertex] (v3) at (0,0) {$v_3$};
\draw (v1) -- (v2);
\draw (v1) -- (v3);

\node at (-1.3,4.1) {$\dom{G}\colon$};
\node[vertexdarkgray] (v12) at (1.5,3) {$v_{1,2}$};
\node[vertexdarkgray] (v13) at (1.5,1) {$v_{1,3}$};
\draw[gray!70] (v1) -- (v12);
\draw[gray!70] (v1) -- (v13);
\draw[gray!70] (v2) -- (v12);
\draw[gray!70] (v3) -- (v13);
\end{tikzpicture}
\caption{A graph, $G$, and its supergraph \dom{G}. The vertices and edges added to $G$ when constructing \dom{G} are drawn in dark gray. The vertex $v_1$ (drawn with a thicker boundary) constitutes an optimal vertex cover for $G$ and an optimal dominating set for \dom{G}.}
\label{fig:dom_example}
\end{figure}

\begin{lemma}\label{lem:optimal_vc_dom}
For any graph, $G$, with no isolated vertices, any optimal vertex cover for $G$ is an optimal dominating set for $\dom{G}$.
\end{lemma}
\begin{proof}
  Let $G=(V,E)$ and let $G'=(V',E')$ be \dom{G}.
  
  We first argue that any vertex cover, $C$, for $G$ is a dominating set for \dom{G}.
  Since every edge in $E$ is covered by $C$ and no vertex in $V$ is isolated, each vertex in $V$ is in $C$ or has a neighbor in $C$.
  Moreover, since each vertex, $v_{i,j} \in V' \setminus V$ has two neighbors which are the two endpoints of an edge in $E$, all vertices not in $V' \setminus V$ have a neighbor in $C$.

  Next, we argue that if $C$ is an optimal vertex cover for $G$, it is an optimal dominating set for \dom{G}.
  This follows from the fact that, for each edge $(v_i,v_j) \in E$, there is a vertex, $v_{i,j} \in V'$ that needs to be in the dominating set or have a neighbor ($v_i$ or $v_j$) in the set. Including $v_i$ or $v_j$ in the set is no worse than including $v_{i,j}$, since each of $v_i$ and $v_j$ dominate all three vertices, $v_i$, $v_j$, and $v_{i,j}$, and $v_{i,j}$ does not dominate any other vertices than these three.
\end{proof}

Note that Lemma~\ref{lem:optimal_vc_dom} does not hold for graphs with isolated vertices, since each isolated vertex would need to be included in a dominating set, but not in a vertex cover.
In an offline reduction, one would simply remove all isolated vertices from the graph, but when receiving vertices online, an algorithm will not be able to tell whether an isolated vertex will stay isolated.
Thus, we give a method for \textquotedblleft connectifying\textquotedblright\ a (possibly) disconnected graph, while controlling the size of the optimal vertex cover.
This is an easy construction that has been considered before~\cite{stackexchange}.
However, we have not found it in refereed works, so for completeness, we give
the argument below.

For any graph, $G=(V,E)$, we let $\con{G} = (V',E')$ be the connected supergraph of $G$, where
\begin{itemize}
\item $V' = V \cup \{s_1,s_2\}$, and
\item $E' = E \cup \{(s_1,s_2)\} \cup \{(v,s_1) \mid v \in V\}$.
\end{itemize}
See Figure~\ref{fig:connectified_graph} for an example of $\con{G}$. 

Observe that given a graph $G$, we can create $\con{G}$ online by first creating vertices $s_1$ and $s_2$, and then revealing all future vertices, $v_i$, as they are being revealed for $G$, together with the additional edge $(s_1,v_i)$.

\begin{figure}[htp]
\centering
\begin{tikzpicture}[scale = 0.8]
\node at (-1,4.5) {$G\colon$};

\node[vertex] (v1) at (0,4.5) {$v_1$};
\node[vertex] (v2) at (4.5,0) {$v_2$};
\node[vertex] (v3) at (4.5,3) {$v_3$};
\node[vertex] (v4) at (3,1.5) {$v_4$};
\node[vertex] (v5) at (1.5,4.5) {$v_5$};
\node[vertex] (v6) at (1.5,3) {$v_6$};

\draw (v2) -- (v3);
\draw (v2) -- (v4);
\draw (v3) -- (v4);
\draw (v3) -- (v6);
\end{tikzpicture}
\hspace{0.5cm}
\begin{tikzpicture}[scale = 0.8]
\node at (-2.2,4.5) {$\con{G}\colon$};

\node[vertexgray] (s1) at (-2,1) {$s_1$};
\node[vertexgray] (s2) at (-2,3) {$s_2$};

\node[vertex] (v1) at (0,4.5) {$v_1$};
\node[vertex] (v2) at (4.5,0) {$v_2$};
\node[vertex] (v3) at (4.5,3) {$v_3$};
\node[vertex] (v4) at (3,1.5) {$v_4$};
\node[vertex] (v5) at (1.5,4.5) {$v_5$};
\node[vertex] (v6) at (1.5,3) {$v_6$};

\draw (v2) -- (v3);
\draw (v2) -- (v4);
\draw (v3) -- (v4);
\draw (v3) -- (v6);

\draw[gray!30] (s1) -- (s2);
\draw[gray!30] (s1) -- (v1);
\draw[gray!30] (s1) -- (v2);
\draw[gray!30] (s1) -- (v3);
\draw[gray!30] (s1) -- (v4);
\draw[gray!30] (s1) -- (v5);
\draw[gray!30] (s1) -- (v6);
\end{tikzpicture}

\caption{Connectifying a graph. The vertices and edges added to $G$ when constructing \con{G} are shown in light gray.}
\label{fig:connectified_graph}
\end{figure}

\begin{lemma}\label{lem:optimal_vc_in_conG}
For any graph, $G$, if $C$ is an optimal vertex cover for $G$,
then $C \cup \{s_1\}$ is an optimal vertex cover for $\con{G}$.
\end{lemma}
\begin{proof}
  If $C$ is a vertex cover for $G$, then $C \cup \{s_1\}$ is a vertex cover for \con{G} since $s_1$ covers all edges added to $G$ when creating \con{G}.
  Moreover, $C \cup \{s_1\}$ is optimal since $(s_1,s_2)$ is not covered by any vertex in $G$ and neither $s_1$ nor $s_2$ covers any edge in $G$.
\end{proof}

\begin{theorem}\label{thm:hardness_dominating_set}
  For any $t \in \ZZ^+$ and any pair of insertion monotone error measures, $\PAIRETA$, the following hold.
  \begin{enumerate}[label = {(\roman*)}]
    \item \DOM on graphs without isolated vertices is $\CCWM{\PAIRETAFORCC}{t}$-hard. \label{thm:dom_no_isolated}
    \item $\DOM$ is asymptotically $\CCWM{\PAIRETAFORCC}{t}$-hard. \label{thm:dom}
  \end{enumerate}
\end{theorem}

\begin{proof}
  We give a strict and an asymptotic reduction from \BDVC{t} to \DOM.
  In the strict reduction, the \BDVC{t} algorithm, \ALG, gives \dom{G} to the \DOM algorithm, $\ALG'$, in an online fashion as it receives the vertices and edges of $G$.
  In the asymptotic reduction, \ALG gives \dom{\con{G}} to $\ALG'$.
  The asymptotic reduction is described in Algorithm~\ref{alg:ORALG_DOM_HARDNESS}.
  The part of Algorithm~\ref{alg:ORALG_DOM_HARDNESS} written in black describes the strict reduction, which can be used if the \BDVC{t} graph is known to contain no isolated vertices.
Figure~\ref{fig:dom_reduction_example} shows an example of the asymptotic reduction.

\begin{algorithm}[h!]
\caption{Reducing \VC to \DOM}
\begin{algorithmic}[1]
\Require $\ALG' \in \ALGS{\DOM}$ and $I = (x,\hat{x},r) \in \INSTANCES{\VC}$
\Ensure $\ALG \in \ALGS{\VC}$ and $I' = (x',\hat{x}',r') \in \INSTANCES{\DOM}$
\State \textcolor{gray}{Give $\ALG'$ a vertex, $s_1$, with prediction $1$ \Comment{Part of \con{G}}}
\State \textcolor{gray}{Give $\ALG'$ a vertex, $s_2$, with prediction $0$ \Comment{Part of \con{G}}}
\State \textcolor{gray}{Give $\ALG'$ the edge $(s_1,s_2)$                  \Comment{Part of \con{G}}}
\State \textcolor{gray}{Give $\ALG'$ a vertex, $s_{1,2}$ with prediction $0$                 \Comment{Part of \dom{\con{G}}}}
\State \textcolor{gray}{Give $\ALG'$ the edges $(s_1,s_{1,2})$ and $(s_{1,2},s_2)$             \Comment{Part of \dom{\con{G}}}}
\For {each vertex, $v_i$, in $I$} \Comment{Step $i$}
    \State \label{line:first} Give $\ALG'$ a vertex, $v_i$, with prediction $\hat{x}_i$ \Comment{Part of $G$}
    \State \textcolor{gray}{Give $\ALG'$ the edge $(s_1,v_i)$        \Comment{Part of \con{G}}}
    \State \label{line:second}\textcolor{gray}{Give $\ALG'$ a vertex, $t_{1,i}$, with prediction $0$      \Comment{Part of \dom{\con{G}}}}
    \State \textcolor{gray}{Give $\ALG'$ the edges $(s_1,t_{1,i})$ and $(t_{1,i},v_i)$   \Comment{Part of \dom{\con{G}}}}
    \For {each edge $(v_j,v_i)$ with $j<i$}
        \State Give $\ALG'$ the edge $(v_j,v_i)$ \Comment{Part of $G$}
        \State \label{line:third} Give $\ALG'$ a vertex, $v_{j,i}$, with prediction $0$     \Comment{Part of \dom{G}}
        \State Give $\ALG'$ the edges $(v_j,v_{j,i})$ and $(v_{i,j},v_i)$  \Comment{Part of \dom{G}}
    \EndFor
    \If {$\ALG'$ accepts at least one vertex in Step $i$} \label{line:vc_to_dom_start_check}
        \State Accept $v_i$
    \EndIf \label{line:vc_to_dom_end_check}
\EndFor
\end{algorithmic}
\label{alg:ORALG_DOM_HARDNESS}
\end{algorithm}

\begin{figure}[htp]
\centering
\begin{tikzpicture}
\node at (-0.8,4.1) {$G\colon$};
\node[vertex,thick] (v1) at (0,2) {$v_1$};
\node[vertex] (v2) at (0,4) {$v_2$};
\node[vertex] (v3) at (0,0) {$v_3$};

\draw (v1) -- (v2);
\draw (v1) -- (v3);
\end{tikzpicture}
\hspace{1.2cm}
\begin{tikzpicture}
\node at (-3.7,4.1) {$\dom{\con{G}}\colon$};
\node[vertexgray,thick] (s1) at (-2,2) {$s_1$};
\node[vertexgray] (s2) at (-2,4) {$s_2$};
\node[vertexdarkgray] (s12) at (-3,3) {$s_{1,2}$};

\draw[gray!30] (s1) -- (s2);
\draw[gray!70] (s1) -- (s12);
\draw[gray!70] (s2) -- (s12);

\node[vertex,thick] (v1) at (2,2) {$v_1$};
\node[vertexdarkgray] (s11) at (1,2.75) {$t_{1,1}$};

\draw[gray!70] (s1) -- (s11);
\draw[gray!70] (v1) -- (s11);

\node[vertex] (v2) at (2,4) {$v_2$};
\node[vertexdarkgray] (s12) at (0,4) {$t_{1,2}$};

\draw[gray!70] (s1) -- (s12);
\draw[gray!70] (v2) -- (s12);

\node[vertex] (v3) at (2,0) {$v_3$};
\node[vertexdarkgray] (s13) at (0,0) {$t_{1,3}$};

\draw[gray!70] (s1) -- (s13);
\draw[gray!70] (v3) -- (s13);

\draw (v1) -- (v2);
\draw (v1) -- (v3);

\draw[gray!30] (s1) -- (v1);
\draw[gray!30] (s1) -- (v2);
\draw[gray!30] (s1) -- (v3);

\node[vertexdarkgray] (v12) at (3.2,3) {$v_{1,2}$};
\draw[gray!70] (v1) -- (v12);
\draw[gray!70] (v2) -- (v12);

\node[vertexdarkgray] (v13) at (3.2,1) {$v_{1,3}$};
\draw[gray!70] (v1) -- (v13);
\draw[gray!70] (v3) -- (v13);
\end{tikzpicture}
\caption{Example reduction for Dominating Set.
  The vertices and edges added to $G$ when constructing \con{G} are drawn in light gray.
  The vertices and edges added to \con{G} when constructing \dom{\con{G}} are drawn in dark gray.
  In $G$, $\{v_1\}$ (drawn with a thicker boundary) is an optimal vertex cover.
  In \dom{\con{G}}, $\{s_1,v_1\}$ (drawn with a thicker boundary) is an optimal dominating set.}
\label{fig:dom_reduction_example}
\end{figure}

Note that for each vertex, $v_i$, in $I$, \ALG accepts $v_i$ only if $\ALG'$ accepts one of the vertices that \ALG gives to it between receiving $v_i$ and $v_{i+1}$.
Therefore, $\ALG(I) \leqslant \ALG'(I')$, satisfying Condition~\ref{item:OR_condition_ALG} of Definition~\ref{def:OR}.

If the strict reduction is used, then, by Lemma~\ref{lem:optimal_vc_dom}, $\OPT(I')=\OPT(I)$, satisfying Condition~\ref{item:OR_condition_OPT}.
Otherwise, by Lemmas~\ref{lem:optimal_vc_dom} and~\ref{lem:optimal_vc_in_conG}, $\OPT(I')=\OPT(I)+1$, satisfying Condition~\eqref{eq:alternative_O2}.

For the strict reduction, Lemma~\ref{lem:optimal_vc_dom} shows that $x$ encodes an optimal dominating set for $I'$.
Since each vertex in \dom{G} corresponding to a vertex, $v_i$, in $G$ have the same prediction as $v_i$, and all other vertices in \dom{G} have prediction~$0$, this means that Condition~\ref{item:OR_condition_eta} is satisfied.
Similarly, for the asymptotic reduction, with the addition that $s_1$ comes with a prediction of $1$ and $s_2$ with a prediction of $0$.
Thus, for this reduction, Condition~\ref{item:OR_condition_eta} follows from Lemmas~\ref{lem:optimal_vc_dom} and~\ref{lem:optimal_vc_in_conG}.
\end{proof}

We summarize our results about $\DOM$ in Theorem~\ref{thm:summary_dom}:

\begin{theorem}\label{thm:summary_dom}
Let $t \in \ZZ^+$.

For any pair of insertion monotone error measures, $\PAIRETA$,
\begin{enumerate}[label = {(\roman*)}]
\item $\DOM$ is asymptotically $\CCWM{\PAIRETAFORCC}{t}$-hard. \label{item:weak_hardness_dom}
\end{enumerate}
For any pair of error measures, $\PAIRETA$,
\begin{enumerate}[label = {(\roman*)}, start = 2]
\item $\DOM \in \CCINFTYWM{\PAIRETAFORCC}$ and $\DOM$ is not $\CCINFTYWM{\PAIRETAFORCC}$-hard. \label{item:dom}
\end{enumerate}
\end{theorem}
\begin{proof}
Item~\ref{item:weak_hardness_dom} follows directly from Theorem~\ref{thm:hardness_dominating_set}.

To see that $\DOM \in \CCINFTYWM{\PAIRETAFORCC}$,
adapt the setup from the proof of the first part of Theorem~\ref{thm:summary_bdvc_t}\ref{item:vc} with the following addition.
Observe that $\ALG'$ creates an infeasible solution to $\DOM$ if there exists a vertex $v_i$ such that $y_i = 0$, and $y_j = 0$, for all vertices $v_j$ for which $(v_i,v_j)$ is contained in the underlying graph of the instance $I$ of $\DOM$.
Since $x$ encodes an optimal dominating set, either $x_i = 1$, or there exists some $j$ for which $(v_j,v_i)$ is contained in the underlying graph of $I$ for which $x_j = 1$, as otherwise $v_i$ is not dominated.
Hence, $\ALG \in \ALGS{\ASGINFTY}$ has guessed $0$ on a true $1$, and so $\ALG(x,\hat{x}) = \infty$.

To see that \DOM is not $\CCINFTYWM{\PAIRETAFORCC}$-hard, note that
the algorithm, $\textsc{Acc} \in \ALGS{\DOM}$, that accepts all vertices is $(n,0,0)$-competitive, as $\OPT_{\DOM}$ has to accept at least one vertex to create a dominating set.
Now, adapt the proof of the last part of Theorem~\ref{thm:summary_bdvc_t}\ref{item:vc}.
\end{proof}

Since allowing the flexibility of the asymptotic online reduction in Definition~\ref{def:asympreduc} is very natural and presumably very useful, it seems natural to define complexity classes based on the asymptotically as hard as relation. This would allow the possibility of complete problems using the asymptotic online reductions. Since the relation is reflexive and transitive, the classes, $\ACCWM{\PAIRETAFORCC}{t}$, can be defined exactly as the $\CCWM{\PAIRETAFORCC}{t}$ classes but using the asymptotically as hard as relation. The properties corresponding to those in Theorem~\ref{lem:different_base_problem}, Corollary~\ref{cor:sub_and_sup_problem}, and Observation~\ref{obs:anyothercanbeused} also hold for these classes, and the class $\ACCWM{\PAIRETAFORCC}{t}$ contains the class $\CCWM{\PAIRETAFORCC}{t}$. In addition, the classes form a hierarchy with $\ASG{t}$ being complete for $\ACCWM{\PAIRETAFORCC}{t}$, giving the same separation of classes as in Theorem~\ref{thm:hierarchy}.

\section{Paging and $\boldsymbol{\CCWM{\PAIRMUFORCC}{t}}$}\label{sec:paging}

In Paging, we have a \emph{universe}, $U$, of $N$ pages, and a \emph{cache} with room for $k < N$ memory pages.
An instance of Paging is a sequence $r = \langle r_1,r_2,\ldots,r_n\rangle$ of requests, where each request holds a page $p \in U$.
If $p$ is not in cache (a \emph{fault}), an algorithm has to place $p$ in cache, either by evicting a page from cache to make room for $p$, or by placing $p$ in an empty slot in the cache, if one exists. 
The cost of the algorithm is the number of faults.
In the following, we assume that the cache is empty at the beginning of the sequence; this affects the analysis by at most an additive constant $k$.

There is a polynomial time optimal offline Paging algorithm, $\LFD$, which first fills up its cache and then, when there is a page fault, always evicts the page from cache whose next request is furthest in the future~\cite{B66}.
Further, it is well-known that no deterministic online Paging algorithm has a competitive ratio better than $k$~\cite{BE98,ST85}.

We study Paging with succinct predictions~\cite{ABEFHLPS23}:

\begin{definition}\label{def:paging}
An instance of \emph{Paging with Discard Predictions} ($\PAG{k}$) is a triple $I = (x,\hat{x},r)$, where $r = \langle r_1,r_2,\ldots,r_n\rangle$ is a sequences of pages from $U$, and $x,\hat{x} \in \{0,1\}^n$ are two bitstrings such that
\begin{align*}
x_i = \begin{cases}
0, &\mbox{if $\LFD$ keeps $r_i$ in cache until it is requested again,} \\
1, &\mbox{if $\LFD$ evicts $r_i$ before it is requested again,}
\end{cases}
\end{align*}
and $\hat{x}_i$ predicts the value of $x_i$.
If the page in $r_i$ is never requested again, then $x_i = 0$ if $\LFD$ keeps the page in cache until all requests have been seen, and $x_i = 1$ otherwise.
\end{definition}

In this section, we consider the pair \PAIRMU of error measures.
Recall from Definition~\ref{def:error_measures_for_ASG_analysis} that $\mu_0$ and $\mu_1$ count the number of wrong predictions of $0$ and $1$, respectively.
We prove that \PAG{t} is contained in  $\CCWM{\PAIRMUFORCC}{t}$, but not $\CCWM{\PAIRMUFORCC}{t}$-hard.

In~\cite{ABEFHLPS23}, Antoniadis et al.\ prove strong lower bounds for online algorithms with predictions for $\PAG{t}$.
Since $\PAG{t} \in \CCWM{\PAIRMUFORCC}{t}$, these bounds extend  to all $\CCWM{\PAIRMUFORCC}{t}$-hard problems.
This as well as upper bounds for \PAG{t} is discussed in Section~\ref{sec:pagingBounds}

\subsection{$\boldsymbol{\PAG{t}}$ is contained in $\boldsymbol{\CCWM{\PAIRMUFORCC}{t}}$}\label{sec:paging_in CCt}

In this section, we give a reduction from \PAG{t} to \ASG{t},
using an unnamed $\PAG{k}$ algorithm from~\cite{ABEFHLPS23} that we call \FLUSHwALLzeros (\FLUSH).
The strategy of \FLUSH is as follows.
Each page in cache has an associated bit which is the prediction bit associated with the latest request for $p$.
Whenever a page not in cache is requested, evict a page from cache whose associated bit is $1$, if such a page exists. 
Otherwise, evict all pages from cache.
We restate a positive result from~\cite{ABEFHLPS23} on the competitiveness of $\FLUSH$ with respect to the pair of error measures $\PAIRMU$.

\begin{theorem}(\cite{ABEFHLPS23})\label{thm:competitiveness_flush}
For any instance, $I = (x,\hat{x},r)$ of $\PAG{k}$, 
\begin{align*}
\FLUSH(I) \leqslant \OPT_{\PAG{k}}(I) + (k-1) \cdot \MUZERO(I) +  \MUONE(I).
\end{align*}
That is, $\FLUSH$ is strictly $(1,k-1,1)$-competitive for $\PAG{k}$ with respect to $\PAIRMU$.
\end{theorem}

We use this result from~\cite{ABEFHLPS23} to establish a relation between Paging and $\ASG{t}$:

\begin{theorem}\label{thm:pag_membership}
For any $t \in \ZZ^+$, $\PAG{t} \in \CCWM{\PAIRMUFORCC}{t}$.
\end{theorem}
\begin{proof}
We define a strict online reduction $\rho \colon \PAG{t} \orarrow \ASG{t}$.
For any instance, $I = (x,\hat{x},r) \in \INSTANCES{\PAG{t}}$, and any algorithm, $\ALG' \in \ALGS{\ASG{t}}$,
we construct an instance, $I' = (x',\hat{x}',r') \in \ASG{t}$, where $x'$ consists of $x$ followed by $t$ $1$'s, and $\hat{x}'$ consists of $\hat{x}$ followed by $t$ $1$'s.
In this way, $\MU_b(I) = \MU_b(I')$, for $b \in \{0,1\}$, so Condition~\ref{item:OR_condition_eta} is satisfied.

\begin{description}
\item[Towards Condition~\ref{item:OR_condition_OPT}:]
Since the cache is empty from the beginning, $\OPT_{\PAG{t}}$ incurs a cost of $1$ on each of the first $t$ requests. 
After this, $\OPT_{\PAG{t}}$ incurs a cost of $1$ for each request with true bit $1$, by the definition of discard predictions.
Thus,
\begin{align}
\label{eq:OPTPag}
  \OPT_{\PAG{t}}(I) = t + \sum_{i=1}^n x_i.
\end{align}
Moreover, by definition of $I'$, we have that $\OPT_{\ASG{t}}(I') = t + \sum_{i=1}^n x_i
= \OPT_{\PAG{t}}(I)$.

\item[Towards Condition~\ref{item:OR_condition_ALG}:]
By definition of $I'$, and letting $y'$ be the output of $\ALG'(I')$,
\begin{align}
  \label{eq:algASGlower}
\ALG'(I') \geqslant \sum_{i=1}^n \big( y'_i + t \cdot x_i \cdot (1-y'_i) \big) + t.
\end{align}
In particular, $\sum_{i=1}^n \big( y'_i + t \cdot x_i \cdot (1-y'_i) \big)$ is the cost of $\ALG'$ on the first $n$ requests, and $t$ is a lower bound on the cost of $\ALG'$ on the last $t$ requests.

We give pseudo-code for $\rho$ given $\ALG' \in \ALGS{\ASG{t}}$ and $I \in \INSTANCES{\PAG{t}}$ in Algorithm~\ref{alg:paging}. 
Formally, $\ALG = \ORALG{\rho}(\ALG')$ runs the algorithm $\FLUSH$ from~\cite{ABEFHLPS23} using $y'_i$ as the prediction for $r_i$.
Hence, letting $y' = \langle y'_1,y'_2,\ldots,y'_n\rangle$ and $I_{y'} = (x,y',r)$, we have that $\ALG(I) = \FLUSH(I_{y'})$.
Therefore,
\begin{align*}
\ALG(I) 
\leqslant \; &\OPT_{\PAG{t}}(I_{y'}) + (t-1) \cdot \MUZERO(I_{y'}) + \MUONE(I_{y'}), \text{ by Theorem~\ref{thm:competitiveness_flush}} \\
= \; &t + \sum_{i=1}^n \left( x_i + (t-1) \cdot  (1-y'_i) \cdot x_i + y'_i \cdot (1-x_i) \right), \text{ by~(\ref{eq:OPTPag})} \\
= \; &t + \sum_{i=1}^n \big( t \cdot x_i \cdot (1-y'_i) + y'_i \big)\\
\leqslant \; &\ALG'(I'), \text{ by~(\ref{eq:algASGlower})}.
\end{align*}
\end{description}
\end{proof}

\begin{algorithm}[h]
\caption{Reducing \PAG{t} to \ASG{t}}
\begin{algorithmic}[1]
  \Require An algorithm $\ALG' \in \ALGS{\ASG{t}}$ and an instance $(x,\hat{x},r) \in \INSTANCES{\PAG{t}}$
\Ensure An $\ASG{t}$ instance $I'=(x',\hat{x}',r')$ and a $\PAG{t}$ algorithm, \ALG
\For{each request $r_i$}
	\State Get prediction $\hat{x}_i$
	\State Ask $\ALG'$ to guess the next bit given $\hat{x}'_i = \hat{x}_i$, and let $y'_i$ be its output
	\If {$r_i$ is a fault}
		\If {there is a page, $p$, in cache with associated bit $1$}
			\State Evict $p$
		\Else
			\State Flush the cache
		\EndIf
		\State Place the page from $r_i$ in cache
	\EndIf
	\State Set the associated bit of $r_i$ to $y'_i$
\EndFor
\For{$i \gets 1$ to $t$}
        \State{Ask $\ALG'$ to guess the next bit given $\hat{x}'_{n+i}=1$}
\EndFor
\State Compute $x \gets \LFD(r)$ and reveal $x'$ which is $x$ appended by $t$ $1$s to $\ALG$
\end{algorithmic}
\label{alg:paging}
\end{algorithm}

\subsubsection{Establishing Bounds with Respect to $\boldsymbol{(\mu_0,\mu_1)}$}\label{sec:pagingBounds}
Now that we have proven that $\PAG{t} \in \CCWM{\PAIRMUFORCC}{t}$, we can extend a collection of lower bounds from Paging with Discard Predictions by Antoniadis et al.~\cite{ABEFHLPS23} to all $\CCWM{\PAIRMUFORCC}{t}$-hard problems:

\begin{theorem}\label{thm:stronger_lower_bounds_from_paging_general_version}
Let $t \in \ZZ^+$, and let $P$ be any $\CCWM{\PAIRMUFORCC}{t}$-hard problem.
Then, for any $\ABC$-competitive algorithm for $P$ with respect to $\PAIRMU$, 
\begin{enumerate}[label = {(\roman*)}]
\item $\ALPHA + \BETA \geqslant t$ and \label{item:lower_bound_condition_1}
\item $\ALPHA + (t-1) \cdot \GAMMA \geqslant t$. \label{item:lower_bound_condition_2}
\end{enumerate}
\end{theorem}
\begin{proof}
By Theorem~\ref{thm:pag_membership}, $\ASG{t} \geq_{\textrm o} \PAG{t}$.
Since $P$ is $\CCWM{\PAIRMUFORCC}{t}$-hard, we have that $P \geq_{\textrm o} \ASG{t}$.
Hence, by transitivity of the as-hard-as relation, $P \geq_{\textrm o} \PAG{t}$.

For~\ref{item:lower_bound_condition_1}, we assume for the sake of contradiction that there is an $\ABC$-competitive algorithm, $\ALG \in \ALGS{P}$, with $\ALPHA + \BETA < t$.
Since $P \geq_{\textrm o} \PAG{t}$, there exists an $(\alpha,\beta,\gamma)$-competitive algorithm, $\ALG_{\PAG{t}}'$, for $\PAG{t}$ with $\alpha + \beta < t$.
This contradicts Theorem 1.7 from~\cite{ABEFHLPS23}.

Similarly, assuming an $\ABC$-competitive algorithm, $\ALG \in \ALGS{P}$, with $\ALPHA + (t-1) \cdot \GAMMA < t$, we obtain a contradiction with Theorem 1.7 from~\cite{ABEFHLPS23} again, thus proving~\ref{item:lower_bound_condition_2}.
\end{proof}

Finally, we re-prove an existing positive result for $\PAG{t}$ (see Remark 3.2 in~\cite{ABEFHLPS23}) but now using that $\PAG{t} \in \CCWM{\PAIRMUFORCC}{t}$:

\begin{theorem}
For any $\alpha,\beta \in \RR^+$ with $\alpha \geqslant 1$ and $\alpha + \beta \geqslant t$, there exists an $(\alpha,\beta,1)$-competitive algorithm for $\PAG{t}$ with respect to $\PAIRMU$.
\end{theorem}
\begin{proof}
By Theorem~\ref{thm:pag_membership}, any $(\alpha,\beta,\gamma)$-competitive algorithm for $\ASG{t}$ implies an $(\alpha,\beta,\gamma)$-competitive algorithm for \PAG{t}.
Thus it is sufficient that, by Theorem~\ref{thm:ftp_equality_wrt_mu_0_and_mu_1}, there exists an $(\alpha,\beta,1)$-competitive algorithm for $\ASG{t}$ with respect to $\PAIRMU$, for any $\alpha,\beta\in\RR^+$ with $\alpha \geqslant 1$ and $\alpha + \beta \geqslant t$. 
\end{proof}

\subsection{$\boldsymbol{\PAG{t}}$ is not $\boldsymbol{\CCWM{\PAIRMUFORCC}{t}}$-hard}
Finally, we show that $\PAG{t}$ is not $\CCWM{\PAIRMUFORCC}{t}$-complete for $t\geqslant 5$.
Since there
is much more information available to an algorithm for $\PAG{t}$, as opposed
to $\ASG{t}$, intuitively, it seems that $\PAG{t}$ should not be as hard as
$\ASG{t}$, so this is the expected result. We prove this by showing that
the lower bound
on $\ASG{t}$ in Lemma~\ref{lem:small_alpha_gives_large_gamma} does not
hold for $\PAG{t}$. That lower bound
shows that for any $\ABC$-competitive algorithm for $\ASG{t}$ with respect to
$\PAIRMU$, if $\alpha<t$, then $\gamma \geqslant 1$. By Observation~\ref{obs:ASG_classes}, we know
that this also applies to any $\CCWM{\PAIRMUFORCC}{t}$-complete problem.
We provide a $\PAG{t}$
algorithm, \Block, that is $(t-\delta,2t,1-\eps)$-competitive for
some positive $\delta$ and $\eps$.

In this section, we say that a page has a prediction of~$0$~($1$), if the latest request to the page had a prediction of~$0$~($1$).

\FlushBlock (\Block) defines consecutive blocks of the input sequence.
A block starts at the beginning of the request sequence or with the next
request after the previous block ends.
Except for the last block, all blocks are complete.
When a block ends, the cache
is flushed.
When evicting a page, \Block considers only pages with a prediction of~$1$ that have not already been evicted within the current block.
Among these pages, it chooses the page that has been in cache longest.
A block ends with a request causing a fault when the cache is full and
one of the following two conditions hold.

\textbf{Condition 1}
All pages in cache have prediction~0.

\textbf{Condition 2}
Some page has prediction~1 and there are no pages with prediction~1 that have not already been evicted in the block.

As a response to the last request of a block, an arbitrary page is evicted.
This decision is not important, since the cache is flushed at the beginning of the next block.

\begin{proposition}\label{obs:cond1}
If a block ends due to Condition 1, the block contains at least one incorrect prediction of $0$.
\end{proposition}
\begin{proof}
  If a block ends due to Condition 1, then
  just before the last request of the block, the cache is filled with pages whose last prediction was~$0$.
  Thus, according to the predictions, \LFD has all of them in cache and will keep them there until their next request.
  This contradicts the fact that the last request of the block is to a page not in cache, so at least one of the predictions must be incorrect.
\end{proof}

The following lemma holds for any block, including the last, possibly incomplete, block.
\begin{lemma}
\label{beta_2k}
For $t\geqslant 3$ and any block, $I_b$, 
$$\Block(I_b) \leqslant \left(t-\frac{1}{t}\right)\LFD(I_b)+2t.$$
\end{lemma}
\begin{proof}
  In any block, \Block only evicts pages that have
  prediction~1 and have not been evicted earlier in the block, possibly except for the page, $p$, evicted as a result of the last request of the block.
  Even if $p$ is evicted twice, there are at most two faults on it within the block, since the block contains no requests after the second eviction of $p$.
  Thus, if there are $s$
  distinct pages in $I_b$, \Block faults at most $2s$ times on $I_b$. \LFD faults at
  least $s-t$ times. If $s-t \leqslant 0$, then \Block has at most $t$
  faults. Otherwise, \Block faults at most $2s$ times and
  $2s =2(s-t)+2t\leqslant 2\LFD(I_b) +2t < (t-\frac{1}{t})\LFD(I_b)+2t$,
  since $t\geqslant 3$.
\end{proof}

\begin{lemma}
  \label{LFDbound}
  For any complete block, $I_b$, with no incorrect $0$-predictions, $$\LFD(I_b) \geqslant 2.$$
\end{lemma}
\begin{proof}
  Let $s$ be the number of distinct pages in the block.
  Since $I_b$ is a complete block, $s>t$.
If $s\geqslant t+2$, then there are at least two pages among the ones requested in
the block that
\LFD does not initially have in cache, so $\LFD(I_b)\geqslant 2$.

Thus, only the case $s=t+1$ remains.
In this case, we let $P$ denote the $t+1$ pages requested within the block $I_b$.
During the first $t$ faults in $I_b$, \Block is filling up the cache again after flushing.
For each of the remaining requests within the block, exactly one of the $t+1$ pages is outside the cache.
We argue that $I_b$ can be split into two subblocks, \subblockOne followed by \subblockTwo, such that each subblock contains requests to all $t+1$ distinct pages in $P$.

Let $r$ be the first request within the block causing an eviction.
At the time $r$ arrives, the cache has been filled and $r$ is a request to the one page in $P$ which is not in cache.
Thus, the part of $I_b$ ending with $r$ contains requests to $t+1$ distinct pages, so we let \subblockOne denote this part of $I_b$.

What remains is to prove that the remaining part, \subblockTwo, of $I_b$ also contains requests to all $t+1$ pages in~$P$.
Note that none of the pages that are in cache at the beginning of \subblockTwo have been evicted from cache within~$I_b$.
We partition $P$ into two subsets, $P_0$ and $P_1$, based on the situation at the beginning of \subblockTwo.
The set $P_0$ consists of the pages whose last request came with a $0$-prediction and $P_1=P \setminus P_0$ consists of those whose last request came with a $1$-prediction.

Note that by the definition of \Block, the last request, $r'$, of the block is a fault and when $r'$ arrives, all pages in cache with a prediction of~$1$ have been evicted within $I_b$. Therefore, all pages in $P_1$ that are in cache at the arrival of $r'$ have been requested with a $0$-prediction or brought into cache 
within \subblockTwo.
If one of the pages in $P_1$ is the page not in cache, it is the page requested by~$r'$.
Thus, all pages in $P_1$ are requested within \subblockTwo.
We now argue that all pages in $P_0$ must also be requested within \subblockTwo.
Since there are no incorrect $0$-predictions in $I_b$ and since any paging algorithm, including \LFD, can have at most $t$ pages in cache at a time, each page in $P_0$ must be requested before the last page in $P_1$ has been requested.

Since \subblockOne and \subblockTwo each contain requests to $t+1$ distinct pages, \LFD must fault at least once in each of the two subblocks.
\end{proof}

\begin{lemma}
\label{rule2}
For $t\geqslant 5$ and any complete block, $I_b$, with no incorrect $0$-predictions,
\[\Block(I_b)\leqslant \left(t-\frac{1}{3t^2}\right)\LFD(I_b)+\left(1-\frac{1}{3t^2}
\right)\MUONE(I_b).\]
\end{lemma}
\begin{proof}
  Let $s$ be the number of distinct pages in the block.
  Since $I_b$ is a complete block, $s>t$.
By Proposition~\ref{obs:cond1},
the block $I_b$ ended due to Condition~2.

\newcommand{\dc}{\ensuremath{d_{\text{c}}}\xspace}
\newcommand{\dw}{\ensuremath{d_{\text{w}}}\xspace}

Let $d$ be the number of pages \Block faults on twice in the block.
Note that for each such page, $p$, the last request to $p$ before the last fault on $p$ in the block has a prediction of $1$.
Among the $d$ pages that \Block faults on twice, let \dc (\dw) be the number of pages, $p$, where the last request to $p$ before the last fault on $p$ had a correct (incorrect) $1$-prediction, and note that $d=\dc+\dw$.

Thus, \Block faults at most $s+d$ times in the block, and
\LFD faults at least $s-t$ times. An independent lower bound on
\LFD is that it faults at least
$\dc$ times. Thus, \Block faults at most
$s+\dw+\dc \leqslant 2\LFD(I_b) + t + \dw$ times.
We now show that for $t \geqslant 5$ and $\eps =\frac{1}{3t^2}$,
\begin{align}
  \label{tobeproven}
  2\LFD(I_b) + t + \dw\leqslant (t-\eps)\LFD(I_b) + (1-\eps)\MUONE(I_b),
\end{align}
which will conclude our proof.

We start by establishing
\begin{align}
  \label{keyIneq}
  t+\eps \dw \leqslant (t-2-\eps)\LFD(I_b).
\end{align}
The easier case is $s=t+1$:
\begin{align*}
  t+\eps \dw
  & \leqslant t+\eps (t+1) \\
  & \leqslant 2t-4-2\eps, \mbox{ for $t\geqslant 5$ and $\eps=\frac{1}{3t^2}$} \\
  & \leqslant (t-2-\eps)\LFD(I_b), \mbox{ since, by Lemma~\ref{LFDbound}, $\LFD(I_b) \geqslant 2$} \\
\end{align*}

So, now we assume that $s \geqslant t+2$:

For $t\geqslant 5$ and $\eps=\frac{1}{3t^2}$,
\[2+\eps t+2\eps \leqslant (t-2-2\eps)(s-t-1),\]
since for $s\geqslant t+2$, the right-hand side is greater than
2.5 and the left-hand side is smaller than~2.5.
This is equivalent to
\begin{align}
  \label{someineq}
  t+\eps (s-t-1+t+1) \leqslant (t-2-\eps)+(t-2-\eps)(s-t-1).
\end{align}
Now,
\begin{align*}
  t+\eps \dw
  & \leqslant t+\eps (s-t-1+t+1), \mbox{ since $\dw \leqslant s$} \\
  & \leqslant (t-2-\eps)+(t-2-\eps)(s-t-1), \mbox{ by (\ref{someineq})} \\
  & \leqslant (t-2-\eps) +(t-2-\eps)(\LFD(I_b)-1), \mbox{ since $\LFD(I_b) \geqslant s-t$} \\
  & = (t-2-\eps)\LFD(I_b)
\end{align*}
Having established (\ref{keyIneq}) for all cases, and adding $2\LFD(I_b) + (1-\eps)\dw$ to both sides, we get that
\begin{align*}
  2\LFD(I_b) + t + \dw
  & \leqslant 2\LFD(I_b) + (t-2-\eps)\LFD(I_b)+(1-\eps)\dw \\
  & \leqslant (t-\eps)\LFD(I_b) + (1-\eps)\MUONE(I_b), \mbox{ since $\MUONE(I_b) \geqslant \dw$}
\end{align*}
This establishes (\ref{tobeproven}) and concludes the proof.
\end{proof}

We can now establish that $\PAG{t}$ is not complete for the class $\CCWM{\PAIRMUFORCC}{t}$.
\begin{theorem}
  \label{blockcomp}
  For $t\geqslant 5$, the $\PAG{t}$ algorithm, \Block, is
  $(t-\frac{1}{3t^2}, 2t, 1-\frac{1}{3t^2})$-competitive.
\end{theorem}
\begin{proof}
  We consider the different types of blocks and show that they are each bounded by
\begin{align}
\left(t-\frac{1}{3t^2}\right)\LFD(I_b)+2t\MUZERO(I_b) +
\left(1-\frac{1}{3t^2}\right)\MUONE(I_b),\label{blockcompineq}
\end{align}
except a possible last block, where we need an additive term of~$2t$.
Thus, we can sum up the faults over all blocks to get the result.

If block $I_b$ ends due to Condition~1 where all pages in cache
have prediction~0, 
then there have been at least $t+1$ different pages requested
in the block, and \LFD cannot have them all in cache.
Thus, there was an incorrect $0$-prediction and
$\MUZERO(I_b) \geqslant 1$.
Using Lemma~\ref{beta_2k}, the number of faults by \Block in $I_b$ is bounded
by
\[\left(t-\frac{1}{t}\right)\LFD(I_b)+2t \leqslant \left(t-\frac{1}{3t^2}\right)\LFD(I_b) + 2t\MUZERO(I_b).\]

If block $I_b$ ends due to Condition~2, then, by Lemma~\ref{rule2},
the number of faults by \Block in $I_b$ is bounded by
\[\left(t-\frac{1}{3t^2}\right)\LFD(I_b)+\left(1-\frac{1}{3t^2}\right)\MUONE(I_b).\]

For the possible final non-ending block, by Lemma~\ref{beta_2k},
the number of faults by \Block in $I_b$ is bounded by
\[\left(t-\frac{1}{t}\right)\LFD(I_b)+2t \leqslant \left(t-\frac{1}{3t^2}\right)\LFD(I_b) + 2t.\]
We have seen that for
any complete block, the number of faults is bounded by (\ref{blockcompineq}).
Thus, they can be summed up to give the result, using
the $2t$ term from the last block as an additive constant.
\end{proof}
\begin{corollary}
  $\PAG{t}$ is not $\CCWM{\PAIRMUFORCC}{t}$-hard.
\end{corollary}
\begin{proof}
  Since \ASG{t} is complete for the class, it follows
  by Lemma~\ref{lem:small_alpha_gives_large_gamma} that
  any algorithm for
  a $\CCWM{\PAIRMUFORCC}{t}$-hard problem
  that has $\alpha<t$, must have $\gamma \geqslant 1$.
  By Theorem~\ref{blockcomp},
  \Block does not fulfill this, so
  $\PAG{t}$ is not $\CCWM{\PAIRMUFORCC}{t}$-hard.
\end{proof}

\section{Concluding Remarks and Future Work}
We have defined complexity classes for online minimization problems with binary predictions and without predictions and proven that they form a strict hierarchy.
Further, we showed that our complexity classes have all the structure one expects from complexity classes.
We proved membership, hardness, and completeness of multiple problems with respect to various pairs of error measures, using our reduction template as well as other methods.
For instance, we have shown completeness of Online $t$-Bounded Degree Vertex Cover and Online $t$-Bounded Overlap Interval Rejection.
Beyond this, we showed strong lower bounds for all $\CCWM{\PAIRMUFORCC}{t}$-hard problems, using similar lower bounds from~\cite{ABEFHLPS23} and a reduction from Paging with Discard Predictions to Asymmetric String Guessing with Predictions. We also proved Paging with Discard Predictions to not be $\CCWM{\PAIRMUFORCC}{t}$-hard.

These results were defined and
proven using the definition of asymptotic competitivenes that allows for
an additive constant.
In online algorithms, researchers generally
use two slightly different
definitions of the asymptotic competitive ratio; either using the style from our paper
or defining the ratio as the
limit of the ratio of the algorithm's performance to \OPT's
performance.
The two definitions are equivalent only if an additive $o(\OPT)$ term is allowed.
Since incorporating this more general term adds a significant amount
of technicalities to the proofs, we have favored readability and just allow
an additive constant.
The difference is minor in that if an algorithm is $c$-competitive with regards
to the limit-based definition, it is $(c+\eps)$-competitive in the
restricted version for any $\eps>0$, so taking the infimum
over all values results in the same competitive ratios.
The results in this paper also hold using a definition of
competitivenes that allows for an additive $o(\OPT)$ term.

Note that our definition of relative hardness also applies to maximization problems.
In a very recent follow-up to the arXiv version of our paper, online maximization problems with predictions were considered and several maximization problems shown to be members, hard, and complete for $\CCWM{\PAIRMUFORCC}{t}$~\cite{B25}. 

The most interesting open problem is extending our complexity theory to other kinds of predictions.
Placing other online problems with or without predictions into the complexity hierarchy,
including determining whether or not 
   $\COLT{k}{t}$ is contained in
any of the complexity classes defined,
would also be interesting.
Other possible directions for future work include 
considering randomized algorithms as well as
other performance measures, for instance using relative worst order analysis~\cite{BF07, BFL07j,BFL21j} or random order analysis~\cite{K96} instead of competitive analysis.

\clearpage
\bibliography{refs}
\bibliographystyle{plain}

\end{document}